\documentclass{article}

\usepackage{amsmath,amssymb,amsthm,bbm,natbib,color,nccmath}
\usepackage[margin=1in]{geometry}
\usepackage{caption}
\usepackage{graphicx}
\usepackage{subcaption}
\usepackage{afterpage}
\usepackage[hang,flushmargin]{footmisc}
\usepackage{float}
\usepackage{enumitem,kantlipsum}
\usepackage{multirow}
\usepackage[ruled,vlined]{algorithm2e}
\usepackage[
  colorlinks,
  citecolor=blue,
  linkcolor=black,
  anchorcolor=red,
  urlcolor=blue
]{hyperref}    

\newtheorem{theorem}{Theorem}

\newtheorem{proposition}{Proposition}

\newtheorem{remark}{Remark}

\allowdisplaybreaks

\newcommand{\dt}{\mathcal{D}_{\textnormal{train}}}
\newcommand{\dc}{\mathcal{D}_{\textnormal{cal}}}
\newcommand{\df}{\mathcal{D}_{f}}
\newcommand{\ba}{\bar{\alpha}}
\newcommand{\ta}{\tilde{\alpha}}
\newcommand{\mt}{m_\textnormal{train}}
\newcommand{\R}{\mathbb{R}}

\newcommand{\D}{\mathcal{D}}
\newcommand{\F}{\mathcal{F}}

\newcommand{\EE}[1]{\mathbb{E}\left[{#1}\right]}
\newcommand{\EEst}[2]{\mathbb{E}\left[{#1}\  \middle| \ {#2}\right]}
\newcommand{\Ep}[2]{\mathbb{E}_{{#1}}\left[{#2}\right]}
\newcommand{\Epst}[3]{\mathbb{E}_{{#1}}\left[{#2}\  \middle| \ {#3}\right]}
\newcommand{\PP}[1]{\mathbb{P}\left\{{#1}\right\}}
\newcommand{\PPst}[2]{\mathbb{P}\left\{{#1}\  \middle| \ {#2}\right\}}

\newcommand{\Pp}[2]{\mathbb{P}_{{#1}}\left\{{#2}\right\}}
\newcommand{\One}[1]{{\mathbbm{1}}\left\{{#1}\right\}}
\newcommand{\one}[1]{{\mathbbm{1}}_{{#1}}}
\newcommand{\iidsim}{\stackrel{\textnormal{iid}}{\sim}}

\newcommand{\X}{\mathcal{X}}
\newcommand{\Y}{\mathcal{Y}}
\newcommand{\ch}{\widehat{C}}
\renewcommand{\hat}{\widehat}
\definecolor{ForestGreen}{RGB}{34,139,34}

\usepackage{authblk}

\allowdisplaybreaks

\title{Conditional predictive inference with $L^k$-coverage control}
\author{Yonghoon Lee}
\author{Zhimei Ren}
\affil{Department of Statistics and Data Science, the Wharton School,\\ University of Pennsylvania}

\date{\today}

\begin{document}

\maketitle

\begin{abstract}
We consider the problem of distribution-free conditional predictive inference. Prior work has established that achieving exact finite-sample control of conditional coverage without distributional assumptions is impossible, in the sense that it necessarily results in trivial prediction sets. While several lines of work have proposed methods targeting relaxed notions of conditional coverage guarantee, the inherent difficulty of the problem typically leads such methods to offer only approximate guarantees or yield less direct interpretations, even with the relaxations. In this work, we propose an inferential target as a relaxed version of conditional predictive inference that is achievable with exact distribution-free finite sample control, while also offering intuitive interpretations. One of the key ideas, though simple, is to view conditional coverage as a function rather than a scalar, and thereby aim to control its function norm. We propose a procedure that controls the $L^k$-norm---while primarily focusing on the $L^2$-norm---of a relaxed notion of conditional coverage, adapting to different approaches depending on the choice of hyperparameter (e.g., local-conditional coverage, smoothed conditional coverage, or conditional coverage for a perturbed sample). We illustrate the performance of our procedure as a tool for conditional predictive inference, through simulations and experiments on a real dataset.

\end{abstract}


\section{Introduction}

Consider a standard predictive inference problem where we are given i.i.d.~data $(X_1,Y_1), \dots, (X_n,Y_n) \in \X \times \Y$. 
The task is to construct, for a new feature input $X_{n+1} \in \X$, 
a prediction set $\ch(X_{n+1})$ for the unobserved corresponding outcome $Y_{n+1}$. 
In distribution-free predictive inference, the goal is to build a prediction set 
whose validity does not rely on the underlying data distribution. 
One of the most widely adopted approaches in this setting is 
conformal prediction~\citep{vovk2005algorithmic}, which guarantees the \textit{marginal coverage}:
\[\Pp{(X_1,Y_1),\cdots,(X_n,Y_n),(X_{n+1},Y_{n+1}) \iidsim P}{Y_{n+1} \in \ch(X_{n+1})} \geq 1 - \alpha, \quad \text{ for any distribution $P$},\]
where $\alpha \in (0,1)$ is a user-specified level. 
The usefulness of the marginal coverage guarantee, however, is often limited, 
since the $(1-\alpha)$ bound pertains to a probability taken not only with respect to the distribution of $Y_{n+1}$, but also with respect to that of $X_{n+1}$ (as well as the training data $\{(X_i,Y_i)\}_{i=1}^n$). 
For instance, in medical prognosis, 
conformal prediction might fail to provide high-quality information for a specific patient,
since its guarantee is averaged over 
a population of hypothetical patients. In conformal inference, the quality of individual-level predictions 
often relies more on the initial point estimate---which can be incorporated into conformal prediction through the nonconformity score---and when this estimate is poor, conformal prediction may fail to provide reliable inference for individuals.

In this context, an important question in the realm of distribution-free predictive inference is how to construct a prediction set with high ``conditional quality". One approach toward this goal is to construct accurate estimators or nonconformity scores to be used within the conformal prediction framework. A representative example is conformalized quantile regression~\citep{romano2019conformalized}, which incorporates estimates of conditional quantiles into conformal prediction. Another line of work has aimed to guarantee conditional quality at the inferential level by directly controlling the conditional coverage rate, rather than relying on the quality of point estimates. An ideal target to consider is the following strong conditional coverage guarantee:
\[\PPst{Y_{n+1} \in \ch(X_{n+1})}{X_{n+1} = x} \geq 1-\alpha, \quad \text{ for all } x \in \mathcal{X}.\]
However, it has been shown that this notion of conditional coverage control is generally unattainable in the distribution-free setting, in the sense that any distribution-free method produces a prediction set of infinite measure whenever the feature distribution is continuous~\citep{vovk2005algorithmic}. Consequently, much recent literature has aimed to achieve a relaxed or asymptotic guarantee, often with additional assumptions.

In this work, we propose a target guarantee that is attainable in an exact, finite-sample, distribution-free sense, while providing conditional coverage control to a meaningful extent. Our target encompasses multiple perspectives, such as a relaxation of the multi-accuracy condition, control of the local-conditional miscoverage rate, and control of the conditional miscoverage rate at a perturbed sample. We develop a method that achieves the proposed target guarantee and demonstrate through empirical results that it effectively controls the conditional miscoverage rates.

\subsection{Notations}
We write $\R$ to denote the real space and $\R^+$ to denote the set of nonnegative real numbers. For a positive integer $n$, $[n]$ denotes the set $\{1,2,\cdots,n\}$. For a measure space 
$(E,\mathcal{E},\mu)$, we define the $L^k$-norm of a measurable function $f$ as
$\|f\|_k := (\int_E |f|^k d\mu)^{1/k}$, and the $L^\infty$-norm of $f$ as
$\|f\|_{\infty} := \inf\{\lambda: |f|\le \lambda \text{ a.e.}\}$.

\subsection{Problem setup}
Suppose we have access to training data $\dt = (X_i',Y_i')_{i \in [n_\text{train}]}$ and calibration data $\dc = (X_i,Y_i)_{i \in [n]}$, drawn i.i.d. from distribution $P_{X,Y} = P_X \times P_{Y  \mid X}$ on $\X \times \Y$. Our objective is to construct a prediction set $\ch(X_{n+1})$ for the unobserved outcome $Y_{n+1}$ upon receiving a new feature input $X_{n+1}$. 

Here, the training data $\dt$ is used to construct a point estimator, e.g., $\hat{\mu}(X_{n+1})$ for $\EEst{Y_{n+1}}{X_{n+1}}$, while the calibration data $\dc$ is then used to construct a prediction set, e.g., of the form $\hat{\mu}(X_{n+1}) \pm c$ for $Y_{n+1}$. Other choices for the form of prediction sets (or, equivalently, the conformity score) are also possible. 

The goal is to achieve a useful conditional coverage guarantee in a distribution-free manner, ensuring that the prediction set has a small conditional miscoverage rate for most values of $X_{n+1}$. 
In our framework, we consider an even stronger notion of the conditional miscoverage rate, given by  
\begin{equation}\label{eqn:cond_misc}
    \alpha_\D(x) = \PPst{Y_{n+1} \notin \ch(X_{n+1})}{X_{n+1} = x, \dt, \dc},
\end{equation}
which conditions on both the {\em test input} and the {\em labeled data $\D := (\dt,\dc)$---i.e., on all the observations we have.} In the above probability, the randomness arises only from the unobserved outcome $Y_{n+1}$, since $\ch(X_{n+1})$ is fully determined by $\dt$, $\dc$, and $X_{n+1}$.
Controlling $\alpha_\D(X_{n+1}) \le \alpha$ almost surely corresponds to achieving both 
training-conditional and test-conditional coverage simultaneously. 
In what follows, we introduce relaxed versions 
of this target, and develop methods that rigorously control them, thereby paving 
the way toward training- and test-conditional coverage guarantees in a practical sense.

\subsection{Overview of results and contributions}

A brief overview of our results and contributions is as follows:\\

\noindent\textbf{Unified formulation for conditional predictive inference.}  
Because of fundamental limitations in distribution-free conditional predictive inference, various approaches have been proposed to target relaxed forms of inference. We provide a unified framework for these approaches through the lens of approximate multiaccuracy. In particular, common targets of conditional predictive inference can be interpreted as controlling function-weighted conditional miscoverage rates for a class of functions.\\ 


\noindent\textbf{Distribution-free inference via probabilistic multiaccuracy with $L^k$-coverage guarantee.}  
We introduce the viewpoint that controlling the conditional miscoverage rate should be seen as controlling a function rather than a single value---making the problem fundamentally different from standard marginal coverage control. Based on this perspective, we propose a method that controls the $L^k$-norm of the function-weighted conditional miscoverage over a predefined class $\F$. This $L^k$-norm control can be interpreted as a form of ``probabilistic multiaccuracy,” ensuring that conditional miscoverage is small for most functions in $\F$. Our procedure achieves exact, finite-sample, distribution-free control at a chosen level $\alpha$, with computational simplicity.\\  

\noindent\textbf{Control of function-weighted conditional miscoverage for broad, interpretable function classes.}  
Our procedure controls function-weighted conditional miscoverage rates for broad function classes $\F$. Our construction only requires $\F$ to be measurable, which allows a wide range of choices while maintaining clear interpretability. For instance, our method can provide control of local conditional coverage, smoothed conditional coverage, conditional coverage for perturbed-sample, and more---depending on the choice of $\F$.\\ 

\noindent\textbf{Empirical evaluation.}  
We empirically assess the proposed procedure through simulations and an application to the Abalone dataset~\citep{nash1994population}.\footnote{Code to reproduce the experiments is available at \url{https://github.com/yhoon31/L_k_conditional_inference}.} The results demonstrate that our method attains tight control of function-weighted conditional miscoverage---which is the theoretical target---while also providing a strong control of feature-conditional miscoverage rate. Additionally, the empirical results show that the proposed procedure does not become unnecessarily conservative to achieve the strong target of conditional miscoverage control---e.g., in `easy' scenarios where standard conformal prediction with marginal coverage already provides good conditional control, the proposed procedure yields a prediction set similar to---i.e. only as conservative as---the conformal prediction set.

\subsection{Related work}

Distribution-free inference has attracted substantial attention in recent literature. Comprehensive overviews of this field are provided by~\citet{vovk2005algorithmic, angelopoulos2021gentle, shafer2008tutorial}, and~\citet{angelopoulos2024theoretical}. A canonical example is conformal prediction~\citep{vovk2005algorithmic, papadopoulos2002inductive}, which offers a general framework for 
predictive inference that achieves marginal coverage guarantees without distributional assumptions. 

Many recent works have attempted to develop methods that go beyond the marginal coverage guarantee---particularly by aiming for test-conditional inference with stronger guarantees, in order to achieve high-quality inference for specific inputs. However, it has been shown that there are fundamental limits to conditional predictive inference in the distribution-free setting. \citet{vovk2005algorithmic} show that strict control of the conditional coverage rate is impossible under a continuous feature distribution, in the sense that any method with distribution-free validity must produce a trivial prediction set. \citet{foygel2021limits} demonstrate the inherent difficulty of conditional predictive inference, even under relaxed inferential targets. \citet{barber2020distribution}, \citet{lee2021distribution}, and \citet{medarametla2021distribution} prove similar impossibility results for the problem of inference on the conditional mean or median.

To address these challenges, various approaches have been proposed. For instance, \citet{romano2019conformalized, chernozhukov2021distributional, guan2023localized},
and~\citet{xie2024boosted}
refine the conformal prediction procedure in ways that can improve conditional performance (empirically), 
while still targeting marginal coverage. 
Other lines of work seek to improve conditional coverage control by pursuing relaxed versions of strict conditional coverage guarantees. 
For example,~\citet{gibbs2025conformal} re-formulate the conditional coverage guarantee as
the multi-accuracy condition, and uses the first-order condition of quantile regression 
to construct prediction sets with approximate local coverage.~\citet{hore2023conformal} 
propose a method that aims for relaxed local coverage guarantees by means of constructing 
prediction set for a randomly perturbed sample from the test point.~\citet{zhang2024posterior}
propose a method that achieves approximate training- and test-conditional coverage guarantee
with the additional assumption that the conformity scores follow a mixture distribution.

On the other hand, a number of works have pursued training-conditional coverage in 
the distribution-free setting
(e.g.,~\citet{vovk2012conditional,bates2021distribution,bian2023training,liang2023algorithmic,gibbs2025characterizing}), but
the guarantees therein are otherwise averaged over the test points.

\section{Inferential target}

We first provide a unified formulation of the problem of distribution-free test-conditional predictive inference, which encompasses different existing approaches aiming at relaxed targets. Based on this formulation, we review existing approaches and then introduce the inferential target we explore in this work.

\subsection{Background: multiaccuracy}\label{sec:multiaccuracy}

To isolate the roles of training-conditional and test-conditional inference, 
we temporarily focus on a fixed function $\ch : \X \rightarrow \Y$ that generates prediction sets,
so that any remaining challenges are solely from ensuring test-conditional coverage.

The following two statements are equivalent:\footnote{The proof of the equivalence is provided in Appendix~\ref{sec:multi_acc}.}
\begin{equation}\label{eqn:multi_acc_ineq}
\begin{split}
    &\text{(i)}\; \alpha(X_{n+1}) := \PPst{Y_{n+1} \notin \ch(X_{n+1})}{X_{n+1}} \leq \alpha, \textnormal{ almost surely}, \\
    &\text{(ii)}\; \EE{f(X_{n+1})\left(\alpha(X_{n+1}) - \alpha\right)} \leq 0, \textnormal{ for all measurable $f$ with $f \geq 0$}.
\end{split}
\end{equation}
Conditions of the form (ii) are often referred to as \textit{multiaccuracy} or \textit{multi-calibration} and has been widely studied in the context of algorithmic fairness~\citep{pmlr-v80-hebert-johnson18a,kim2019multiaccuracy}.
However, constructing a prediction set with the above strong conditional coverage guarantee in a distribution-free manner is known to be impossible~\citep{vovk2005algorithmic}---in the sense that any procedure achieving this guarantee for any distribution produces an infinite-width prediction set. Thus, we consider a plausible relaxation of the above guarantee. First, observe that condition (ii) can be rewritten as
\[\ba(f) := \EE{\frac{f(X_{n+1})}{\EE{f(X)}} \cdot \alpha(X_{n+1})} \leq \alpha, \textnormal{ for all measurable $f$ with $f \geq 0$},\]
with the convention $0/0 = 0$.
A natural relaxation of the above condition is to require the inequality to hold for functions in a space $\F$ instead of all measurable functions:
\begin{equation}\label{eqn:target_general}
    \ba(f) \leq \alpha, \textnormal{ for all $f \in \F$}.
\end{equation}
This type of approach was previously studied in~\citet{gibbs2025conformal},\footnote{In~\citet{gibbs2025conformal}, the authors consider the equality-based condition $\EE{f(X_{n+1}) \cdot (\alpha(X_{n+1})-\alpha)} = 0, \textnormal{ for all $f \in \F$}$.} where they provide methods for $\F$ being the space of linear functions, as well as some approximate guarantee results for reproducing kernel Hilbert spaces and Lipschitz functions.

In fact, it turns out that most inferential targets considered in the literature on conditional predictive inference follow the form of~\eqref{eqn:target_general}. Below, we provide a brief overview of conditional coverage guarantees studied in prior work on distribution-free inference, rephrased in our notation for clarity. 

\paragraph{Strong control of test-conditional coverage.} A natural target for inference with conditional validity is the following condition:
\[\PPst{Y_{n+1} \in \ch(X_{n+1})}{X_{n+1} = x} \geq 1-\alpha, \quad \text{ for almost all } x \in \mathcal{X}.\]
This can equivalently be written as $\alpha(x) := \PPst{Y_{n+1} \notin \ch(X_{n+1})}{X_{n+1}=x} \leq \alpha$ for almost all $x \in \mathcal{X}$.
From the perspective of multiaccuracy, this guarantee can also be viewed as condition~\eqref{eqn:target_general} with the space of delta functions $\F = \{\delta_z : z \in \X\}$---where $\delta_z$ satisfies $\EE{h(X) \delta_z(X)} = h(z)$ for any integrable function $h$---by viewing $\alpha(x) = \EE{\delta_{x}(X_{n+1})\cdot\alpha(X_{n+1})} = \ba(\delta_x)$.

As mentioned earlier, it has been shown that this guarantee is unachievable in a distribution-free manner, in the sense that any prediction set satisfying this condition must have infinite measure. Consequently, several works have explored the possibility of achieving inference under a relaxed version of this target. 

\paragraph{Approximate conditional coverage guarantee.}
\citet{foygel2021limits} considers the following approximate conditional coverage guarantee (under the setting $\X = \R^d$):
\[\PPst{Y_{n+1} \in \ch(X_{n+1})}{X_{n+1} \in B} \geq 1 - \alpha \quad \text{for all } B \subset \R^d \text{ with } P_X(B) \geq \delta,\]
which controls the subset-conditional coverage with any subset of the feature space with probability mass at least $\delta$. 
We can express this condition as
\begin{equation}\label{eqn:approx_CC}
    \ba(f) \leq \alpha \quad \text{for all } f \in \F, \quad \text{where } \F = \{\one{B} : B \subset \R^d, P_X(B) \geq \delta\}.
\end{equation}
To see this, observe that for any set $B \subset \X$ with $P_X(X)\ge \delta$, we have
\begin{equation}\label{eqn:alpha_f_ind}
\begin{split}
    &\ba(\one{B}) = \EE{\frac{\one{B}(X_{n+1})}{\EE{\one{B}(X_{n+1})}} \cdot \alpha(X_{n+1})} = \frac{\EE{\One{X_{n+1} \in B}\cdot \alpha(X_{n+1})}}{\PP{X_{n+1} \in B}}\\
    &= \frac{\EE{\EEst{\One{X_{n+1} \in B} \cdot \One{Y_{n+1} \notin C(X_{n+1})}}{X_{n+1}}}}{\PP{X_{n+1} \in B}} = \frac{\PP{X_{n+1} \in B, Y_{n+1} \notin C(X_{n+1})}}{\PP{X_{n+1} \in B}}\\
    &= \PPst{Y_{n+1} \notin C(X_{n+1})}{X_{n+1} \in B}.
\end{split}
\end{equation}
However, \citet{foygel2021limits} shows that even for this relaxed target, a meaningful distribution-free inference is impossible, i.e., any prediction set with this guarantee must have measure at least as large as that from a trivial method---which is conservative.

\paragraph{Binning-based approach}

A further relaxation of the subset-conditional target---attainable in a distribution-free manner---is to partition the feature space into finitely many bins 
and require coverage within the bins.
Specifically, let $\{\X_1, \X_2, \ldots, \X_K\}$ be a fixed partition of the feature space $\X$. By applying conformal prediction within each bin and combining the resulting score bounds, we obtain a prediction set with the following guarantee:
\begin{equation}\label{eqn:coverage_bin}
    \PPst{Y_{n+1} \in \ch(X_{n+1})}{X_{n+1} \in \X_k} \geq 1-\alpha\qquad\text{ for all } k \in [K] \text{ with } \PP{X_{n+1} \in \X_k} > 0.
\end{equation}
See, e.g.,~\citet{angelopoulos2024theoretical} for details. Note that, compared to the condition~\eqref{eqn:approx_CC}, this is a much weaker target, as it requires control over only a finite number of subsets. The guarantee in~\eqref{eqn:coverage_bin} can be equivalently expressed as follows.
\[\ba(f) \leq \alpha \quad \text{for all } f \in \F, \quad \text{where } \F = \{\one{\X_k} : k \in [K]\}.\]
This approach requires each bin to contain a sufficient number of observations to 
yield informative prediction sets.
Consequently, with a small or moderate total sample size, 
a fine partition may not lead to useful inference.

\paragraph{Coverage conditional on a perturbed sample.}
\citet{hore2023conformal} propose a method that begins by drawing a perturbed sample $\tilde{X}_{n+1}$ around $X_{n+1}$. 
Their procedure achieves
\begin{equation}\label{eqn:randomized}
\EEst{\alpha(X_{n+1},\tilde{X}_{n+1})}{\tilde{X}_{n+1}} \leq \alpha \qquad \text{ almost surely},
\end{equation}
where $\alpha(X_{n+1}, \tilde{X}_{n+1})$ is defined analogously to $\alpha(X_{n+1})$, but with additional conditioning on $\tilde{X}_{n+1}$. Following the discussion in Section~\ref{sec:kernel_2}, this condition can be expressed as
\begin{equation}\label{eqn:randomized_2}
\ba(f_{\tilde{X}_{n+1}}) \leq \alpha \quad \text{almost surely},
\end{equation}
where $f_{\tilde{X}_{n+1}}$ denotes the kernel function centered at $\tilde{X}_{n+1}$. 
This bound controls $\ba(f)$ for a single function $f_{\tilde{X}_{n+1}}$, but it can still be useful since this function is randomized and incorporates the information from $X_{n+1}$.
Yet, the condition~\eqref{eqn:randomized} suffers from a lack of straightforward interpretation as test-conditional inference. Consequently, in~\citet{hore2023conformal}, the authors take additional steps to achieve coverage conditional on $X_{n+1}$ rather than $\tilde{X}_{n+1}$, and derive approximate coverage results for local-conditional coverage. 
\citet{zhang2024posterior} adopt a similar strategy to  
construct prediction sets that are valid conditional on a perturbed version of the test sample;
this guarantee is then connected with the strict conditional coverage 
under the additional assumption that the nonconformity scores 
are from a mixture model.

\paragraph{Direct multiaccuracy-inspired approach.} \citet{gibbs2025conformal} introduces methods 
directly motivated by the condition
\[\EE{f(X_{n+1})(\mathbbm{1}\{Y_{n+1} \in \ch(X_{n+1})\} - (1-\alpha))} = 0 \quad \text{for all } f \in \F,\]
where $\F$ is a prespecified collection of functions. This condition is equivalent to requiring $\ba(f) = \alpha$ for all $f \in \F$. 
Their proposed method attains the marginal coverage guarantee; 
furthermore, when $\F$ is a finite class of linear functions over $d$-dimensional basis functions, 
they show that 
\[
\EE{f(X_{n+1})(\mathbbm{1}\{Y_{n+1} \in \ch(X_{n+1})\} - (1-\alpha))} \ge 0, \text{ for all } f\in \F.
\]
This can be equivalently written as  
\begin{equation}\label{eqn:gibbs}
\ba(f) \le  \alpha,
\qquad  \text{ for all } f \in \F.
\end{equation}

In the special case where $\F = \{x \mapsto \sum_{G \in \mathcal{G}} \beta_G \One{x \in G}\}$ for some finite collection of subsets $\mathcal{G}$, the prediction set controls the group-conditional coverage rate: $\mathbb{P}\{Y_{n+1} \in \ch(X_{n+1}) \mid X_{n+1} \in G\} \geq 1-\alpha,\; \forall G \in \mathcal{G}$.

For more general function classes, they provide a corresponding error bound under mild regularity conditions on $\F$, with further discussion on the case where $\F$ is a reproducing kernel Hilbert space or a space of Lipschitz functions.\\

To summarize, there is a fundamental limitation in directly controlling the conditional coverage $\alpha(X_{n+1})$ under general feature distributions. Intuitively, this is because when $X$ is continuous, the values ${X_i}$ in the training data are almost surely distinct from the test point $X_{n+1}$, providing no information about the spread of $Y_{n+1} \mid X_{n+1}$. To address this issue, a surrogate target $\ba(f)$ is often adopted in various contexts. Although the relaxation~\eqref{eqn:approx_CC}, which requires strict control of $\ba(f)$ for multiple functions $f$, remains unattainable in general, the guarantee~\eqref{eqn:randomized_2} for a single randomized $f$ is achievable. Moreover, in specific cases—such as when $\F$ is a set of linear functions---it is possible to simultaneously control $\ba(f)$ in the sense of~\eqref{eqn:gibbs}, although the interpretation with respect to the specific 
function class is less clear.\\

Returning to the formulation~\eqref{eqn:target_general}, we observe that 
existing approaches remain limited: they either lack clear interpretability or restrict the choice of function space, thereby excluding many theoretically or practically useful cases. 
This work aims to provide a general methodology for predictive inference that controls the ``$f$-weighted conditional miscoverage" $\ba(f)$ for broad classes of function spaces $\F$, while offering simple and interpretable inferential guarantees.

    

\subsection{A new inferential target: probabilistic multiaccuracy via \texorpdfstring{$L^k$}{}-control}\label{sec:target}

We are now ready to introduce the target guarantee that we will prove to be achievable in the distribution-free sense for a broad class of function spaces. Intuitively, we consider a ``probabilistic multiaccuracy" condition of the form
\[\ba(f) = \EEst{\frac{f(X_{n+1})}{\EEst{f(X)}{f}} \cdot \alpha(X_{n+1})}{f} \textnormal{ is small with high probability, for $f$ drawn from $P_f$ on $\F$},\]
for a predefined measurable function space $\F$ and a distribution $P_f$ on $\F$ (the choice of 
$\F$ and $P_f$ will be discussed shortly). 
As a formal theoretical guarantee with the above property, we consider the following \textit{$L^k$-coverage guarantee}:
\[\Ep{f \sim P_f}{\ba(f)^k} \leq \alpha^k, \textnormal{ or equivalently, } \|\ba(f)\|_k = \Ep{f \sim P_f}{\ba(f)^k}^{1/k} \leq \alpha,\]
which bounds the $L^k$ distance between the function $f \mapsto \ba(f)$ and the zero function. Intuitively, the above conditions for $k=1,2,\dots$ form a family of guarantees where larger values of $k$ impose a stronger condition on $\ba(f)$—the condition with $k=1$ only controls the expectation $\EE{\ba(f)}$ by $\alpha$, allowing large $\ba(f)$ values for some proportion of $f \in \F$, while the condition with large $k$ essentially approximates $\|\ba(f)\|_\infty \leq \alpha$, or equivalently, $\ba(f) \leq \alpha$ for all $f$ in the support of $P_f$. We provide more details for the interpretation of this condition in Section~\ref{sec:moment}.

So far, our discussion is based on the simplifying assumption that $\hat C$ is a fixed function.
Recall that, in reality, the prediction set function $\ch$ is a random object constructed using the datasets $\dt$ and $\dc$. Thus, the formal statement of the guarantee we aim for is
\begin{equation}\label{eqn:k_moment_condition_0}
    \|\ba_\D(f)\|_k \leq \alpha, \text{ where } \ba_\D(f) = \EEst{\frac{f(X_{n+1})}{\EEst{f(X)}{f}} \cdot \alpha_\D(X_{n+1})}{f, \dt, \dc},
\end{equation}
where $\alpha_\D(\cdot)$ is defined in~\eqref{eqn:cond_misc}. 

Intuitively, we aim to construct a prediction set whose $f$-weighted conditional miscoverage rate $\ba_\D(f)$ is controlled in the $L^k$ norm sense over the space $\F$---which, in turn, implies that $\ba(f)$ is small for most functions in $\F$. As a remark, even when considering the strong notion of conditional coverage, which conditions not only on $f$ but also on the data $\dt$ and $\dc$, the main difficulty still arises mostly from conditioning on $f$---at a high level, this is because the i.i.d. datasets do not introduce much variability to the procedure $\ch$ and its resulting miscoverage rate, whereas the component $f$ exerts a more significant influence, with higher variance in the miscoverage rate.


\paragraph{The normalizing factor.}
A remaining (relatively minor) challenge in constructing a distribution-free achievable target is that the weight $f(x)/\EEst{f(X)}{f}$ requires the normalizing constant $\EEst{f(X)}{f}$ for exact computation, which is not available in practice.
Thus, we consider the following guarantee with the `data-driven normalization' instead:
\begin{equation}\label{eqn:k_moment_condition}
    \|\ta_\D(f)\|_k \leq \alpha, \text{ where } \ta_\D(f) = \EEst{\frac{f(X_{n+1})}{\widehat{\EEst{f(X)}{f}}} \cdot \alpha_\D(X_{n+1})}{f, \dt, \dc},
\end{equation}
where the estimator $\widehat{\EEst{f(X)}{f}}$ is constructed using $\dt$---more precisely, $(X_i')_{i \in [n_\text{train}]}$. We will discuss later a reasonable choice for this estimator. For now, we note that obtaining an accurate estimator of this term using $(X_i')_{i \in [n_\text{train}]}$ for each $f$ is a relatively straightforward task, as it amounts to estimating a mean from i.i.d. data, in contrast to the more involved steps that follow for inference.
Although we introduce the notion $\ta_\D(f)$ and the condition~\eqref{eqn:k_moment_condition} for theoretical completeness, one may view the condition~\eqref{eqn:k_moment_condition_0} as the essential target, with the term $\EEst{f(X)}{f}$ being numerically computed in practice.

As a final remark, we highlight that the only restriction on the function space $\F$ in this formulation is measurability, so that a distribution $P_f$ can be defined on it. In Section~\ref{sec:method}, we provide a methodology that achieves the above guarantee for any prespecified measurable function space $\F$.

\subsection{Interpretation of the function-weighted conditional miscoverage control}\label{sec:interpretations}

The target guarantee~\eqref{eqn:k_moment_condition_0} states that the prediction set has small $\ba_\D(f)$ for ``most'' functions $f \in \F$. But what does it actually mean to have small $f$-weighted conditional miscoverage rates? Is this merely a mathematical approximation of the multiaccuracy condition~\eqref{eqn:multi_acc_ineq}?

In fact, more detailed and clear interpretations can be attained, depending on the  choice of the (measurable) function space $\mathcal{F}$. Below, we present different interpretations along with a few examples of function spaces. For simplicity, we temporarily disregard the data $\mathcal{D}$ and focus on what $\ba(f)$---with the subscript $\mathcal{D}$ omitted---represents, and we denote the prediction set as $C(X_{n+1})$ and the conditional miscoverage rate at $X_{n+1}=x$ as $\alpha(x)$.

\paragraph{Controlling the local miscoverage rate.}\label{sec:local}

We begin with a simple case where $\mathcal{F}$ consists of indicator functions, e.g.,  
\[\mathcal{F} = \{\one{B(x,r)} : x \in \mathcal{X}\}, \quad \text{where } B(x,r) = \{x' \in \mathcal{X} : d(x,x') \leq r\},\]  
where $d$ is a metric defined on $\mathcal{X}$ and $r>0$ is a predefined radius. For this space of neighborhoods of possible feature values, we have $\ba(\one{B}) = \PPst{Y_{n+1} \notin \ch(X_{n+1})}{d(X_{n+1},x) \leq r}$, which represents the miscoverage rate conditional on the event that $X_{n+1}$ lies within a neighborhood of the point $x$---by applying arguments similar to~\eqref{eqn:alpha_f_ind}. The target guarantee aims to keep this local-conditional miscoverage rate small for ``most'' $x \in \mathcal{X}$; 
one may also consider the class of ball-indicator functions with varying radii.

\paragraph{Controlling the smoothed conditional miscoverage rate.}\label{sec:kernel}

Suppose the function space consists of kernel functions:  
\begin{equation}\label{eqn:kernel}
    \mathcal{F} = \{f_z : z \in \mathcal{X} \}, \quad \text{ where } f_z : x \mapsto K(z,x), 
\end{equation}
for a kernel function $K$, e.g., the Gaussian kernel $K(z,x) = \exp(-\frac{1}{2h^2}\|z - x\|^2)$ with a bandwidth $h > 0$. Intuitively, the function $f_z$ can be seen as a surrogate for the delta function $\delta_z$ (Recall the discussion in Section~\ref{sec:multiaccuracy}). Our target guarantee aims to satisfy
\[\ba(f_z) = \EE{\tilde{f}_z(X_{n+1}) \cdot \alpha(X_{n+1})} \text{ is small for ``most'' } z \in \mathcal{X},\]  
where $\tilde{f}_z(x) = f_z(x) / \EE{f_z(X)}$ serves as the surrogate for $\delta_z$, 
obtained by normalizing $f_z$. Intuitively, $\ba(f_z)$ can be viewed as a kernel-smoothed version of the conditional miscoverage rate $\alpha(z) = \ba(\delta_z)$. One may vary the bandwidth parameter $h$
in the function class, and posit a distribution over it when defining a distribution $P_f$ on $\mathcal{F}$.

\paragraph{Controlling the conditional miscoverage rate for perturbed sample.}\label{sec:kernel_2}
An alternative interpretation of $\ba(f_z)$ under the setting~\eqref{eqn:kernel} is to view it as the conditional miscoverage rate at a perturbed sample. Observe that, by definition, for each $z \in \mathcal{X}$, we have  
\[\EE{\tilde{f}_z(X)} = 1 = \int \tilde{f}_z(x) \cdot f_X(x) \,\mathrm{d}x,\]  
where $f_X$ denotes the density function of $P_X$, assuming that $P_X$ has a well-defined density function. Therefore, $\tilde{f}_z \cdot f_X$ is also a valid density function, and now we let $\tilde{Z}$ be a sample drawn from its corresponding distribution, denoted $\tilde{f}_z^X$.
Intuitively, $\tilde{Z}$ is drawn around the point $z$ since its distribution $\tilde{f}_z^X$ is based on the kernel function $f_z$, but the distribution is reweighted by $f_X$. For example, if $f_X(x) \approx 0$, then $x$ is unlikely to be drawn as the value of $\tilde{Z}$, even if it is near $z$. On the other hand, if $f_X(x)$ is relatively large, then $x$ is more likely to be drawn among the values whose distance from $z$ is similar. In summary, $\tilde{Z}$ can be viewed as a perturbation of $z$, with its distribution weighted by $f_X$.

Having drawn a sample $\tilde{Z}$ from $\tilde{f}_z^X$, we can write $\ba(f_z)$ as
\[\ba(f_z) = \EE{\tilde{f}_z(X_{n+1}) \cdot \alpha(X_{n+1})} = \int \tilde{f}_z(x) \alpha(x) f_X(x) \, \mathrm{d}x = \Ep{\tilde{Z} \sim \tilde{f}_z^X}{\alpha(\tilde{Z})}.\]
Therefore, by changing the notation from $\tilde{Z}$ to $\tilde{X}_{n+1}$ for better interpretability, and by viewing $z$ as a realization of $X_{n+1}$, $\ba(f_z)$ can be represented as
\[\EEst{\alpha(\tilde{X}_{n+1})}{X_{n+1} = z} = \PPst{\tilde{Y}_{n+1} \notin \ch(\tilde{X}_{n+1})}{X_{n+1} = z}.\]
which denotes the miscoverage rate for the perturbed test sample $(\tilde{X}_{n+1},\tilde{Y}_{n+1})$, drawn as $\tilde{X}_{n+1} \mid X_{n+1} \sim \tilde{f}_{X_{n+1}}^X$ and $\tilde{Y}_{n+1} \mid \tilde{X}_{n+1} \sim P_{Y \mid X}$. 

As a remark, a similar interpretation appears in~\citet{gibbs2025conformal}, where the authors provide an interpretation in the sense of covariate shift defined by $f$. We choose to interpret it as a perturbation rather than a shift, as we are primarily interested in---and our method accommodates---the setting where $\mathcal{F}$ is set as a collection of kernel functions, e.g.,~\eqref{eqn:kernel}.

\paragraph{Surrogate of the strict multiaccuracy condition.}

Finally, following the motivation outlined in Section~\ref{sec:target}, one can simply consider a large space of diverse functions to better approximate the space of all measurable functions and the multiaccuracy condition~\eqref{eqn:multi_acc_ineq}. For example, one can set  
\[\mathcal{F} = \{\text{linear functions}\} \cup \{\text{Gaussian kernels}\} \cup \{\text{uniform kernels}\},\]  
and define a distribution $P_f$ as a mixture of distributions, each defined on a component of the union. 

As a remark, in practical settings where a certain level of smoothness in the conditional miscoverage rates can be expected, including kernel-type functions may suffice for a useful approximation. This is because we primarily need to approximate the `space of delta functions,' as described in the section~\ref{sec:kernel}, whereas the space of all measurable functions can generally be an unnecessarily conservative formulation.

\subsection{The \texorpdfstring{$L^k$} --coverage guarantee}\label{sec:moment}
In Section~\ref{sec:interpretations}, we set aside the randomness in 
$\hat C$ and focus on the interpretation from the test-conditional perspective. 
In this section, we take the training randomness back into consideration 
and give a holistic interpretation of the \texorpdfstring{$L^k$} --coverage guarantee.

We first note that controlling the conditional coverage rate is fundamentally different from controlling the marginal coverage rate, in the sense that it requires controlling a function rather than a single value. Specifically, for the standard target of controlling the conditional miscoverage rate $\alpha(X_{n+1})$, the goal is essentially to control the function
\[x \mapsto \alpha_\D(x) = \PPst{Y_{n+1} \notin \ch(X_{n+1})}{X_{n+1}=x, \D},\]
Similarly, in our problem, the aim is to control the function
\[f \mapsto \ba_\D(f) = \EEst{\tilde{f}(X_{n+1}) \cdot \alpha(X_{n+1})}{f, \D}.\]
For such goals, a natural target is to bound the ``size" of the miscoverage rate function---$\alpha(\cdot)$ or $\ba(\cdot)$. Alternatively, one can view the problem as controlling a random variable---which is mathematically a measurable function---such as $\alpha_\D(X_{n+1})$ or $\ba(f)$.

For instance, one may aim to control the $L^1$ norm by a predetermined level $\alpha$, meaning that the target guarantee can be represented as
\[\|\alpha_\D(X_{n+1})\|_1 = \EE{\alpha_\D(X_{n+1})} = \PP{Y_{n+1} \notin \ch(X_{n+1})} \leq \alpha,\]
which is equivalent to the marginal coverage guarantee. 
Of course, controlling only the mean of the target function might be considered weak, and one might consider the $L^\infty$ norm instead:
\[\|\alpha_\D(X_{n+1})\|_\infty = \sup_{x \in \X} |\alpha_\D(x)| \leq \alpha, \quad \text{ i.e., } \quad \alpha_\D(x) \leq \alpha \text{ for all } x \in \X,\]
which is generally impossible to achieve, as discussed previously. Generally, one may consider controlling the $L^k$ distance, in which case the target guarantee is given by $\|\alpha_\D(X_{n+1})\|_k \leq \alpha$. 
Note that this condition serves as a stronger target for larger $k$, since $\|Z\|_p \leq \|Z\|_q$ generally holds for $1 \leq p \leq q$ and a random variable $Z$. 

Similarly, to control $f \mapsto \ba_\D(f)$, we generally consider the $L^k$ condition $\|\ba_\D(f)\|_k \leq \alpha$, as discussed in Section~\ref{sec:target}.

\paragraph{Controlling $\|\alpha_\D(X_{n+1})\|_k$ versus $\|\tilde \alpha_\D(f)\|_k$.}
Generally, in the setting of i.i.d.~data from a continuous distribution, achieving the $L^k$-coverage condition beyond $k=1$---the marginal coverage guarantee---is hard. Intuitively, this is because we have only one $Y$ sample for each $X$ sample, making it difficult to learn beyond the first moment. Attaining a useful prediction set with the $L^k$-coverage guarantee for $k \geq 2$ may be possible if the dataset somehow has a more informative structure. For example, \citet{lee2023distribution} provide a methodology that ensures the second-moment coverage guarantee $\|\alpha_\D(X_{n+1})\|_2 \leq \alpha$ in cases where the outcome $Y$ is drawn repeatedly for each $X$ sample.

On the other hand, for the alternative target $f \mapsto \ba_\D(f)$, it turns out that the $L^k$-coverage guarantee for $k \geq 2$ is achievable in the i.i.d.~setting---as we discuss in Section~\ref{sec:method}. The key idea is to artificially introduce a repeated measurement structure in the dataset by first drawing samples of functions $f$ and then assigning $k$ pairs of $(X, Y)$ samples to each drawn function. The detailed steps for constructing a prediction set that satisfies the $L^k$-coverage condition are provided in the next section.

\paragraph{Implications on the tail bounds.} 
One direct implication of the $L^k$-coverage guarantee is a bound on the tails of 
the training- and test-conditional coverage: for any $\epsilon \in (0,1)$, 
we have 
\begin{align}\label{eqn:tailbound}
\PP{\ba_\D(f) \ge \epsilon} \le \frac{\EE{\ba_\D(f)^k}}{\epsilon^k} 
\le \frac{\alpha^k}{\epsilon^k} = (\alpha / \epsilon)^k.
\end{align}
For example, when $\alpha = 0.1, \epsilon = 0.2$, and $k=2$, the upper bound in~\eqref{eqn:tailbound} 
evaluates to $0.25$. In other words, with probability at least $0.75$, the conditional coverage is at least 
$80\%$ {\em given the training and test input}. The upper bound quickly drops as $K$
increases; if $K=3$, the tail probability becomes $0.125$.

\section{Inference with \texorpdfstring{$L^k$}{}-conditional coverage control}\label{sec:method}

\subsection{Conditional predictive inference with \texorpdfstring{$L^2$}{}-coverage guarantee}
Now, we present a method for constructing the prediction set $\ch(X_{n+1})$ that satisfies the guarantee~\eqref{eqn:k_moment_condition}. For conciseness, we outline the procedure for the case of \( k=2 \)---which turns out to be a particularly useful choice in practice, as we will discuss later. The extension to larger values of $k$ is straightforward and is deferred to Appendix~\ref{sec:method_k}.

We first fix the function space $\F$ of interest and a distribution $P_f$ on $\F$, 
where we assume the functions in $\F$ are bounded with $b = \sup_{f \in \F} \|f\|_\infty$. 
Then, letting $m = \lfloor n/2 \rfloor$, we draw 
\[f_1, f_2, \dots, f_m \iidsim P_f.\]
Let $\df = (f_i)_{i \in [m]}$ represent the set of sampled functions.

Next, we construct a nonconformity score $s: \X \times \Y \rightarrow \R^+$ using the training data $\dt$, and then consider a family of prediction set functions $\{C_t : t \geq 0\}$, where $C_t(x) = \{y \in \Y : s(x,y) \leq t\}$. We now need to find a $\hat{t}$ such that the prediction set $C_{\hat{t}}(X_{n+1})$ satisfies the target guarantee~\eqref{eqn:k_moment_condition}. 
For each $t \geq 0$ and $i \in [n]$, let $Z_i^t = \One{Y_i \notin C_t(X_i)}$. Now, using the training data $\dt$, we `normalize' the functions $f_1,\cdots,f_m$. Specifically, let $\mt = \lfloor n_\text{train}/2 \rfloor$, and define  
\begin{equation}\label{eqn:normalizing}
    \tilde{f}_i = f_i/\gamma(f_i) \text{ for $i \in [m]$, where }\gamma(f) = \left(\frac{1}{\mt+1}\left( \sum_{l=1}^{\mt} f(X_{2l-1}') f(X_{2l}') + b^2\right)\right)^{1/2}.
\end{equation}

We then construct the prediction set $\ch(X_{n+1})$ based on the following formula:

\begin{equation}\label{eqn:chat}
\begin{split}
    &\ch(x) = C_{\hat{t}}(x) = \{y \in \Y : s(x,y) \leq \hat{t}\}, \text{ where }\\
    &\hat{t} = \min\left\{t \geq 0 : \frac{\sum_{i=1}^m  \tilde{f}_i(X_{2i-1})\tilde{f}_i(X_{2i})Z_{2i-1}^tZ_{2i}^t+\frac{b^2}{\gamma(f)^2}}{\sum_{i=1}^m \tilde{f}_i(X_{2i-1})\tilde{f}_i(X_{2i})} \leq \alpha^2\right\},
\end{split}
\end{equation}
where $f$ is an additional sample from $\F$ independently drawn from $P_f$. 
Here, the minimum in the definition of $\hat{t}$ is well-defined, as the variables $Z_i^t$ are discrete, taking values in $\{0,1\}$. We define $\hat{t} = \infty$ if there is no $t\geq 0$ that satisfies the condition. The overall procedure is summarized in Algorithm~\ref{alg:main}.

\begin{algorithm}[h!]
\caption{Conditional predictive inference with $L^2$-coverage guarantee}
\label{alg:main}
\KwIn{Training data $\dt =(X_i',Y_i')_{i \in [n_\text{train}]}$, 
calibration data $\dc = (X_i,Y_i)_{i \in [n]}$, test input $X_{n+1}$, 
target level $\alpha$, sampling distribution of functions $P_f$, with a uniform bound $b$.}

\vskip 2ex

{\bf Step 1:} Construct the nonconformity score function $s : \X \times \Y\rightarrow \R^+$ using the training data $\dt$.
\vskip 1ex
{\bf Step 2:} Draw a sample of functions $f_1,f_2,\cdots,f_m \iidsim P_f$ and $f \sim P_f$, where $m = \lfloor n/2 \rfloor$.
\vskip 1ex
{\bf Step 3:} Compute the normalizing constants $(\gamma(f_i))_{i \in [m]}$ and $\gamma(f)$, based on~\eqref{eqn:normalizing}.
\vskip 1ex
{\bf Step 4:} Compute the threshold $\hat{t}$ as
\[\hat{t} = \min\left\{t \geq 0 : \frac{\sum_{i=1}^m  \tilde{f}_i(X_{2i-1})\tilde{f}_i(X_{2i})Z_{2i-1}^tZ_{2i}^t+\frac{b^2}{\gamma(f)^2}}{\sum_{i=1}^m \tilde{f}_i(X_{2i-1})\tilde{f}_i(X_{2i})} \leq \alpha^2\right\},\]
\hspace{13mm} where $Z_i^t = \One{s(X_i,Y_i) > t}$, for $i \in [n]$.
\vskip 2ex 

\KwOut{Prediction set $\ch(X_{n+1}) = \{y \in \Y : s(X_{n+1},y) \leq \hat{t}\}$.}
\end{algorithm}

\begin{remark}[Randomness]
We note that although the additional sample $f$ introduces some randomness into the procedure, its effect remains insignificant, as the term $b^2/\gamma(f)^2$ is dominated by the summation of $m$ terms, $\sum_{i=1}^m \tilde{f}_i(X_{2i-1})\tilde{f}_i(X_{2i})Z_{2i-1}^tZ_{2i}^t$---unless the sample size is too small. We further note that, as we detail below, the conditional coverage rate we consider is defined by conditioning on all sources of randomness, including $\dt$, $\dc$, $\X_{n+1}$, as well as $\df$ and $f$---ensuring that the resulting prediction set maintains reliability for the specific observations we have and the specific function samples drawn. 

To obtain a procedure that is not randomized by $f$, one can alternatively define a more conservative threshold $\hat{t}$ as
\[\hat{t} = \min\left\{t \geq 0 : \frac{\sum_{i=1}^m \tilde{f}_i(X_{2i-1})\tilde{f}_i(X_{2i})Z_{2i-1}^tZ_{2i}^t + \frac{b^2}{\min_{f \in \F}\gamma(f)^2}}{\sum_{i=1}^m \tilde{f}_i(X_{2i-1})\tilde{f}_i(X_{2i})} \leq \alpha^2\right\},\]
where in practice the term $\min_{f \in \F}\gamma(f)^2$ typically needs to be computed approximately---e.g., from large samples of functions in $\F$. However, this approach can result in conservative prediction sets in several cases, and thus we present the procedure~\eqref{eqn:chat} as our main recommendation.
\end{remark}

\subsection{Theoretical guarantees}
Now, we prove the coverage guarantee of the above prediction set. We show that the resulting guarantee controls the following notion of conditional miscoverage rate, which conditions on all sources of randomness:

\begin{equation}\label{eqn:cond_misc_f}
    \alpha_\D(x) = \PPst{Y_{n+1} \notin \ch(X_{n+1})}{X_{n+1} = x, \D}, \text{ where } \D = (\dt, \dc, \df ,f).
\end{equation}
This is essentially equivalent to the conditional probability $\mathbb{P}\{Y_{n+1} \notin \ch(X_{n+1}) \mid X_{n+1} = x, \ch\}$, where the only source of randomness is $Y_{n+1}$. The formal theoretical 
guarantee is stated as follows.
\begin{theorem}\label{thm:main}
    The prediction set $\ch(X_{n+1})$ given as~\eqref{eqn:chat} satisfies the $L^2$-coverage guarantee, i.e., 
    \[\|\ta_\D(f)\|_2 \leq \alpha, \text{ where } \ta_\D(f) = \EEst{\frac{f(X_{n+1})}{\gamma(f)} \cdot \alpha_\D(X_{n+1})}{\D},\]
    where $\alpha_\D$ denotes the conditional miscoverage rate, defined as~\eqref{eqn:cond_misc_f}.
\end{theorem}
The proof of Theorem~\ref{thm:main} is deferred to Appendix~\ref{appd:proof_main}. 
Several remarks are in order.

\begin{remark}
The choice of the normalizing constant $\gamma(f)$ in~\eqref{eqn:normalizing} was made to ensure  
\[\gamma(f)^2 = \frac{1}{\mt+1} \left( \sum_{l=1}^{\mt} f(X_{2l-1}') f(X_{2l}') + b^2 \right) \approx \Epst{X,X' \iidsim P_X}{f(X)f(X')}{f} = \EEst{f(X)}{f}^2.\]
We choose $\gamma(f)$ as the variable whose square accurately approximates $\EEst{f(X)}{f}^2$, rather than choosing it to approximate $\EEst{f(X)}{f}$ directly. This is because the objective of the approximation is the guarantee $\EE{\ba_\D(f)^2} \leq \alpha^2$ rather than $\ba_\D(f)$ itself---observe that
\[\EE{\ba_\D(f)^2} = \EE{\EEst{\frac{f(X_{n+1})}{\EEst{f(X)}{f}} \cdot \alpha_\D(X_{n+1})}{\D}^2} = \EE{\frac{1}{\EEst{f(X)}{f}^2}\cdot\EEst{f(X_{n+1})\cdot\alpha_\D(X_{n+1})}{\D}^2}.\]

\end{remark}

\begin{remark}[Tightness]
The equality in the $L^2$-coverage condition, $\|\ta_\D(f)\|^2 \leq \alpha^2$, is approximately attained in the `ideal' case where $\alpha_\D(X_{n+1}) \equiv \alpha$, since in this setting, we have
\begin{multline*}
\EE{\ta_\D(f)^2} = \alpha^2 \cdot \EE{f(X_{n+1})f(X_{n+2}) / \gamma(f)^2} = \alpha^2 \cdot \EE{\frac{(\mt+1)\cdot f(X_{n+1})f(X_{n+2})}{\sum_{l=1}^{\mt} f(X_{2l-1}') f(X_{2l}') + b^2}}\\
\lesssim \alpha^2 \cdot \EE{\frac{(\mt+1)\cdot f(X_{n+1})f(X_{n+2})}{\sum_{l=1}^{\mt} f(X_{2l-1}') f(X_{2l}') + f(X_{n+1})f(X_{n+2})}} = \alpha^2,
\end{multline*}
where $X_{n+2} \sim P_X$ is a hypothetical sample drawn independently of $\D$ and $X_{n+1}$, and the last equality follows from the exchangeability of the pairs $(X_1', X_2'), \dots, (X_{2\mt-1}',X_{2\mt}'), (X_{n+1}, X_{n+2})$.
\end{remark}

\begin{remark}[Alternative choice of $\hat t$]
    The threshold $\hat{t}$ in~\eqref{eqn:chat} can alternatively be defined as
    \[\hat{t} = \min\left\{t \geq 0 : \frac{1}{m+1}\cdot \left(\sum_{i=1}^m  \tilde{f}_i(X_{2i-1})\tilde{f}_i(X_{2i})Z_{2i-1}^tZ_{2i}^t+\frac{b^2}{\gamma(f)^2}\right) \leq \alpha^2\right\}.\]
    The steps in the proof of Theorem~\ref{thm:main} ensure that the prediction set $\ch(X_{n+1}) = \{y \in \Y : s(x,y) \leq \hat{t}\}$ with this choice of $\hat{t}$ satisfies the same coverage guarantee. In practice, the two definitions of $\hat{t}$ yield nearly identical results, and the difference between them is typically negligible. Nevertheless, we adopt the definition in~\eqref{eqn:chat} as the primary formulation, as it can, at least intuitively, account for certain biases that may be present in the training data.
\end{remark}

\subsection{Choice of \texorpdfstring{$P_f$}{}}\label{sec:P_f_choice}

The proposed procedure~\ref{alg:main} requires a predefined function-sampling distribution $P_f$. While $P_f$ can, in principle, be any distribution over the measurable function space $\mathcal{F}$, it may also be chosen in a manner that adapts to the unknown feature distribution $P_X$. To provide some intuition, we present examples below in the context of the kernel function space $\mathcal{F} = \{f_z : z \in \mathcal{X}\}$, where $f_z : x \mapsto K(z, x)$. In this setting, defining the function-sampling distribution $P_f$ reduces to specifying a distribution for a random variable $Z$ over $\mathcal{X}$ (see Section~\ref{sec:kernel}).

\paragraph{Procedure with a prespecified \texorpdfstring{$P_f$}{}.}

Assuming that the feature space $\mathcal{X}$ is known---which is typically the case---one can define a distribution $P_Z$ over $\mathcal{X}$, sample $Z_1, \dots, Z_m \overset{\text{i.i.d.}}{\sim} P_Z$, and then set $f_i = f_{Z_i}$ for each $i \in [m]$. For example, if one aims to ensure uniformly good conditional inference across the feature space, $P_Z$ can be chosen as $\text{Unif}(\mathcal{X})$---in the case where the uniform distribution is well-defined on $\mathcal{X}$. Alternatively, one may assign greater weight to a specific region of interest in the feature space, for example by selecting a bell-shaped distribution centered around that region or by defining a uniform distribution over a ball contained within it.

\paragraph{Procedure with unlabeled feature samples.}
Suppose we have a set of unlabeled test data points $\tilde{X}_{1}, \cdots, \tilde{X}_{\tilde{n}}$. For example, if a group of test points $X_{n+1}, \cdots, X_{n+n_\text{test}}$ is observed instead of a single one, we may treat $X_{n+2}, \cdots, X_{n+n_\text{test}}$ as an unlabeled set that we can use for the purpose of making inference on $X_{n+1}$. In this case, the function samples can be set directly as $f_i = f_{\tilde{X}_i}$ for $i \in [m]$ (when $\tilde{n} \geq m = \lfloor n/2\rfloor$), 
which naturally assigns more weight to regions of the feature space with higher likelihood. As a remark, this strategy naturally addresses the covariate shift setting in a sense, as it allows the test points to be directly used as kernel centers.

\subsection{Extension}

The underlying idea behind procedure~\ref{alg:main} is to first sample functions $(f_i)_{i \in [m]}$ and then use two $(X,Y)$ observations to learn---roughly speaking---the conditional distribution of $Y$ given $X$, weighted or perturbed by $f$. Here, the sample size used to cover the function space $\F$ is $m = \lfloor n/2 \rfloor$, while for each sampled $f$, the sample size for learning the $f$-conditional distribution is only two. This might result in inefficient use of the data—for example, in a setting where $f$ is a uniform kernel with a relatively small bandwidth, it is possible that for most $f_i$'s, the corresponding two observations both fall outside the ball, generating a noninformative prediction set. Can we make a different choice in this tradeoff between the resources used for covering sufficient functions and those used for learning the $f$-conditional distributions? The answer is positive, as we explain below.

Suppose now that for a positive integer $r \geq 2$, we draw $m = \lfloor n/r \rfloor$ functions $f_1, \dots, f_m \iidsim P_f$ and also $f \sim P_f$, and then assign $r$ points to each $f_i$, denoted as $(X_{k,j}, Y_{k,j})_{j \in [r]}$, where $X_{k,j} = X_{(k-1)r+j}$ and $Y_{k,j} = Y_{(k-1)r+j}$ for $k \in [m]$ and $j \in [r]$. Then we define
\begin{equation}\label{eqn:chat_r}
\begin{split}
    &\ch(x) = \{y \in \Y : s(x,y) \leq \hat{t}\}, \text{ where }\\
    &\hat{t} = \min\left\{t \geq 0 : \frac{\sum_{k=1}^m \frac{1}{\binom{r}{2}} \sum_{1 \leq j_1 < j_2 \leq r} \tilde{f}_k(X_{k,j_1})\tilde{f}_k(X_{k,j_2})Z_{k,j_1}^t Z_{k,j_2}^t+\frac{b^2}{\gamma(f)^2}}{\sum_{k=1}^m \frac{1}{\binom{r}{2}} \sum_{1 \leq j_1 < j_2 \leq r} \tilde{f}_k(X_{k,j_1})\tilde{f}_k(X_{k,j_2})} \leq \alpha^2\right\},
\end{split}
\end{equation}
where $\tilde{f}_k$ is defined as in~\eqref{eqn:normalizing}, and $Z_{k,j}^t = \One{s(X_{k,j},Y_{k,j}) > t}$ for $k \in [m], j \in [r]$, and $t \geq 0$.
The following theorem proves the validity of the above procedure.
\begin{theorem}\label{thm:generalized}
    The prediction set $\ch(X_{n+1})$ given as~\eqref{eqn:chat_r} satisfies the $L^2$-coverage guarantee $\|\ta_\D(f)\|_2 \leq \alpha$, where $\ta_\D(f)$ is defined as in Theorem~\ref{thm:main}.
\end{theorem}

Although this extension also attains the same guarantee, we generally recommend the main procedure in Algorithm~\ref{alg:main} with an appropriate choice of the function space (e.g., Gaussian kernels rather than uniform kernels), in order to maximize the effective sample size $m = \lfloor n/r \rfloor$ with $r=2$ and consequently to avoid potentially conservative inference. As we illustrate through experiments in Section~\ref{sec:sim}, assigning two points to each function sample---i.e., applying the main procedure---is sufficient to tightly attain the $L^2$-coverage guarantee.

\subsection{Remark on \texorpdfstring{$L^k$}{}-coverage guarantee with a larger \texorpdfstring{$k$}{}}
In general, the $L^k$-conditional coverage guarantee~\eqref{eqn:k_moment_condition} can be attained through a similar approach, with further details provided in Appendix~\ref{sec:method_k}. However, attaining this guarantee for larger $k$ values
while maintaining prediction sets of practical width requires a substantially large sample size. 
Therefore, in many practical problems where finite-sample inference is desired, we believe that the procedure with an $L^2$-coverage guarantee is the most useful---indeed, as demonstrated in our experiments, the $L^2$-coverage already offers strong performance in controlling conditional miscoverage rates.

On the other hand, if we in fact have access to a large number of samples, one may aim for a stronger guarantee---such as the $L^5$-coverage guarantee. Heuristically, with an infinite sample size, our method can closely approximate the $L^\infty$ guarantee, $\ba(f) \leq \alpha$ for all $f \in \F$, for instance, by applying the $L^{100}$-coverage guarantee---as long as $\F$ admits a finite measure. Note the difference from the result for the guarantee in~\eqref{eqn:approx_CC}, where it has been shown that any prediction set satisfying the guarantee must have infinite width, even with an unlimited amount of data~\citep{foygel2021limits}.

Therefore, roughly speaking, achieving $\ba(f) \leq \alpha$ for all $f \in \F$ is not impossible in a distribution-free sense---for any measurable function class $\F$---, although it yields meaningful inference only in the limit of infinitely many data points. For example, $\F$ may be taken as the set of all kernel functions with varying bandwidths, ensuring that at least one $\ba(f)$ closely approximates the feature-conditional miscoverage $\alpha(x)$, except in pathological cases where $\alpha(x)$ possesses atypical or irregular structure---since with infinitely many data points, this allows accurate inference on $\ba(f)$ for all $f \in \F$, even for kernels with very small bandwidths.

While the above arguments offer a heuristic theoretical perspective on the possibility for meaningful distribution-free conditional predictive inference with a general $L^k$-coverage guarantee---approximating the $L^\infty$ guarantee---we again recommend, from a methodological perspective, the procedure based on the $L^2$-coverage guarantee because of the reasons above. Additionally, we emphasize that the $L^2$ coverage guarantee need not be viewed merely as a surrogate or relaxation of the $L^\infty$ guarantee---it serves as a meaningful conditional-miscoverage-controller in its own right, avoiding the overly strict or conservative nature of the $L^\infty$ guarantee.

\section{Simulations}\label{sec:sim}
To demonstrate the behavior of the prediction set~\eqref{eqn:chat}, we provide a series of experimental results. To provide an intuitive illustration of the $L^2$-coverage guarantee, we begin with a simple example involving a one-dimensional feature under extreme settings, showcasing the behavior of the proposed method in different scenarios. For comparison, we also provide the results from the standard split conformal prediction, which guarantees marginal, i.e., $L^1$-coverage, guarantee. We then present experimental results in a more complex setting with multivariate features and a nontrivial outcome-feature dependence structure.

\subsection{Illustration of the \texorpdfstring{$L^2$}--coverage guarantee}

We consider the following two outcome distributions:\footnote{The experimental design in this section is motivated by the simulation settings in~\citet{lee2023distribution}.} 
\begin{enumerate}
    \item Setting 1: $Y \mid X \sim \mathcal{N}(0.5 X + 0.1 X^2, 3^2)$,
    \item Setting 2: $Y \mid X \sim \mathcal{N}(0.5 X + 0.1 X^2, \sigma(X)^2)$, where $\sigma(X) = \One{X \leq 8} + 5\cdot \One{X > 8}$,
\end{enumerate}
where in both settings, the feature is drawn from the uniform distribution $X \sim \textnormal{Unif}([0,10])$. A visual representation using scatter plots is provided in Figure~\ref{fig:sim_scatter}.

\begin{figure}[ht]
    \centering
    \includegraphics[width=0.75\linewidth]{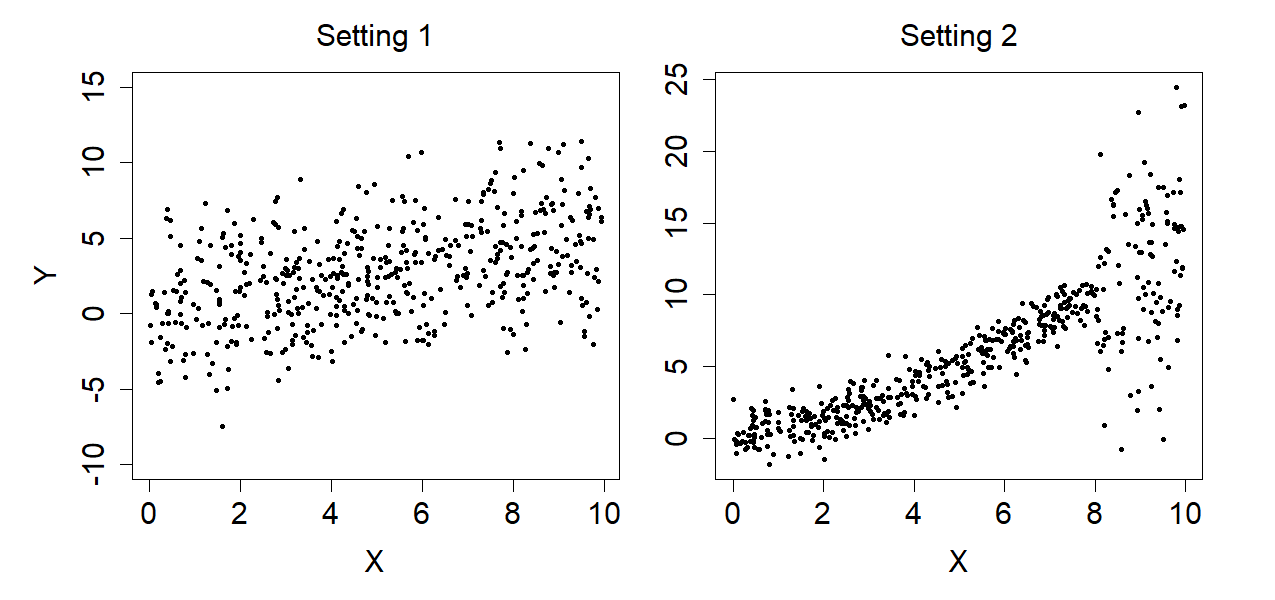}
    \caption{Scatter plots of the samples generated under Setting 1 and Setting 2.}
    \label{fig:sim_scatter}
\end{figure}

Setting 1 represents a case where the conditional variance of $Y$ given $X$ is constant, implying that the difficulty of conditional predictive inference is uniform across the feature space. In contrast, Setting 2 illustrates a scenario in which a subset of the feature space is associated with high conditional variance, making conditional predictive inference more challenging.

We compare the behavior of the split conformal prediction method and the proposed procedure using a function class consisting of Gaussian kernels with bandwidth $h = \sqrt{2}$:
\[\mathcal{F} = \{f_z : z \in \mathcal{X}\}, \quad \text{where } f_z(x) = \exp\left(-(z - x)^2/4\right),\]
and the distribution $P_f$ over $\mathcal{F}$ is defined as the distribution of $f_Z$ where $Z \sim P_X$. For the nonconformity score function, we consider the following two choices:
\begin{enumerate}
    \item \textbf{Absolute residual score:} $s(x, y) = |y - \hat{\mu}(x)|$, where $\hat{\mu}$ is fitted via linear regression.
    \item \textbf{Quantile-based score:} $s(x, y) = \max\{\hat{q}_{\alpha/2}(x) - y,\; y - \hat{q}_{1 - \alpha/2}(x)\}$, where $\hat{q}_{\alpha/2}$ and $\hat{q}_{1 - \alpha/2}$ are fitted using random forest quantile regression.
\end{enumerate}
These two choices serve as (extreme) examples of a ``bad score'' and a ``good score'', respectively, to illustrate the behavior of the procedure and its coverage guarantee under different score qualities. Specifically, the first score reflects the case of model misspecification, as well as the failure to account for heteroscedasticity in the spread of $Y$ across different values of $X$ in Setting 2, whereas the second score offers more accuracy and flexibility in this regard. We set the level as $\alpha=0.2$ in the following experiments.

We generate a training dataset of size $n_{\text{train}} = 500$ and construct the nonconformity scores. Then, we repeat the following steps 500 times: generate a calibration dataset of size $n=2000$, construct the prediction sets using the split conformal prediction and our proposed method based on Algorithm~\ref{alg:main}, and compute the conditional miscoverage rate $\alpha_\D(x)$ at a randomly drawn test point $X_{n+1}$. We note that $\alpha_\D(x)$ can be explicitly computed when the conditional distribution of $Y$ given $X$ is known, as it essentially corresponds to the conditional probability $\mathbb{P}\{Y_{n+1} \notin \ch(X_{n+1}) \mid \ch(X_{n+1})\}$. Note also that, based on the discussion in Section~\ref{sec:kernel}, $\ba_\D(f_x)$ (and consequently $\tilde \alpha_D(f_x)$) approximates $\alpha_\D(x)$, 
and therefore, the proposed method approximately controls $\|\alpha_\D(X)\|_2$ by controlling $\|\ba_\D(f)\|_2$. The results are shown in Figure~\ref{fig:sim_residual}, ~\ref{fig:sim_quantile}, and~\ref{fig:sim_width}.

\begin{figure}[htbp]
    \centering
    \includegraphics[width=0.85\linewidth]{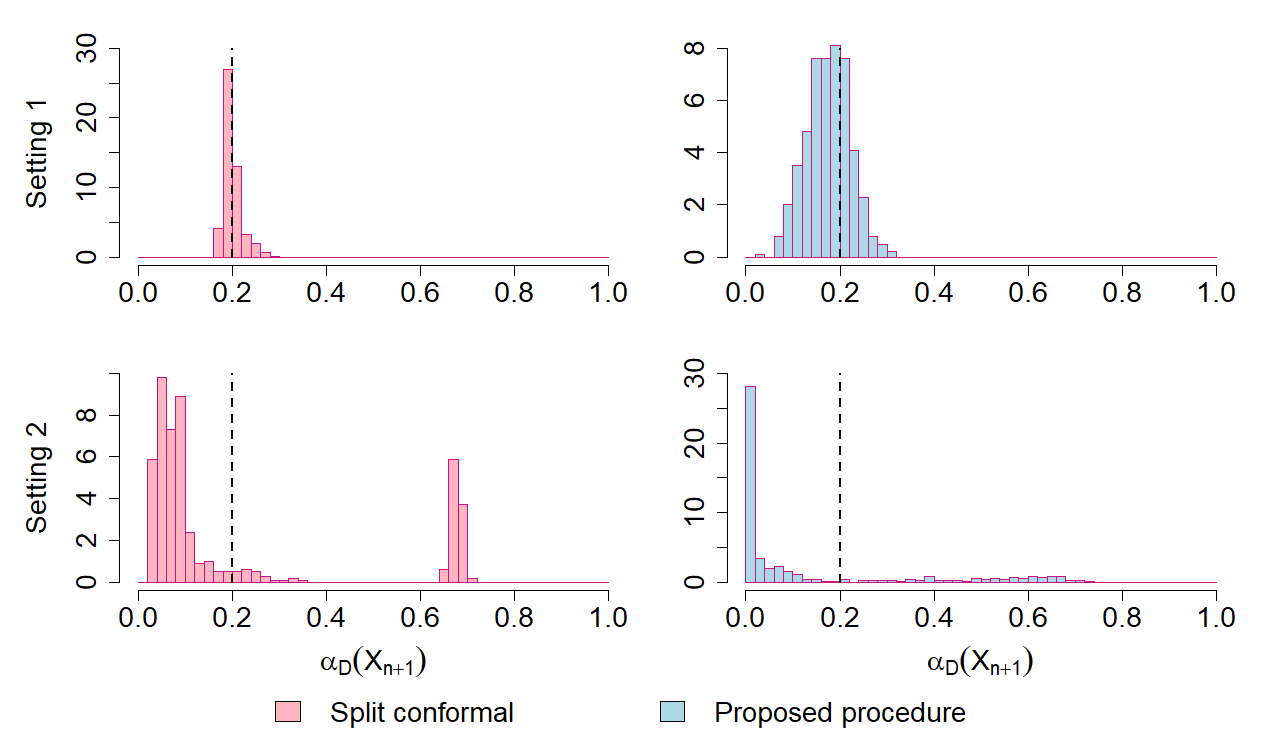}
    \caption{Conditional miscoverage rates of the prediction sets constructed using split conformal prediction and the proposed procedure with absolute residual score, under Settings 1 and 2. The dotted line corresponds to $\alpha=0.2$.}
    \label{fig:sim_residual}
\end{figure}

\begin{figure}[htbp]
    \centering
    \includegraphics[width=0.85\linewidth]{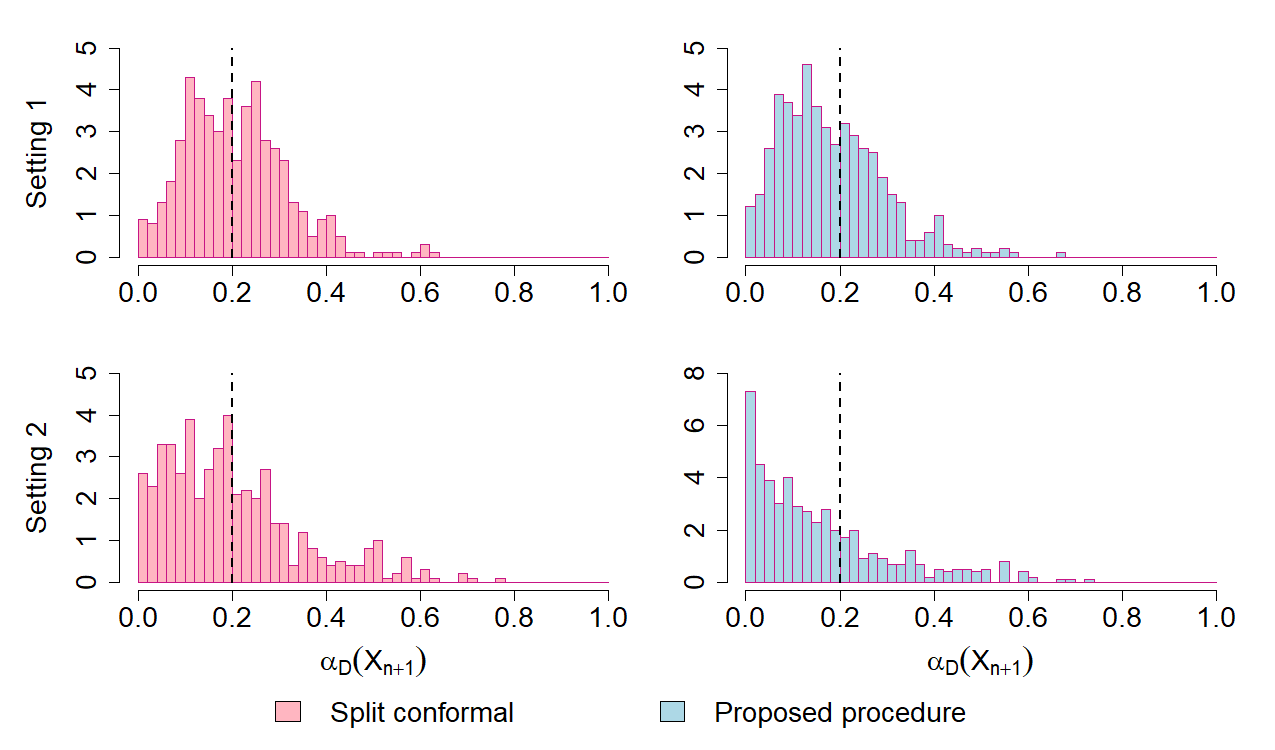}
    \caption{Conditional miscoverage rates of the prediction sets constructed using split conformal prediction and the proposed procedure with quantile-based score, under Settings 1 and 2. The dotted line corresponds to $\alpha=0.2$.}
    \label{fig:sim_quantile}
\end{figure}

\begin{figure}[htbp]
    \centering
    \includegraphics[width=0.95\linewidth]{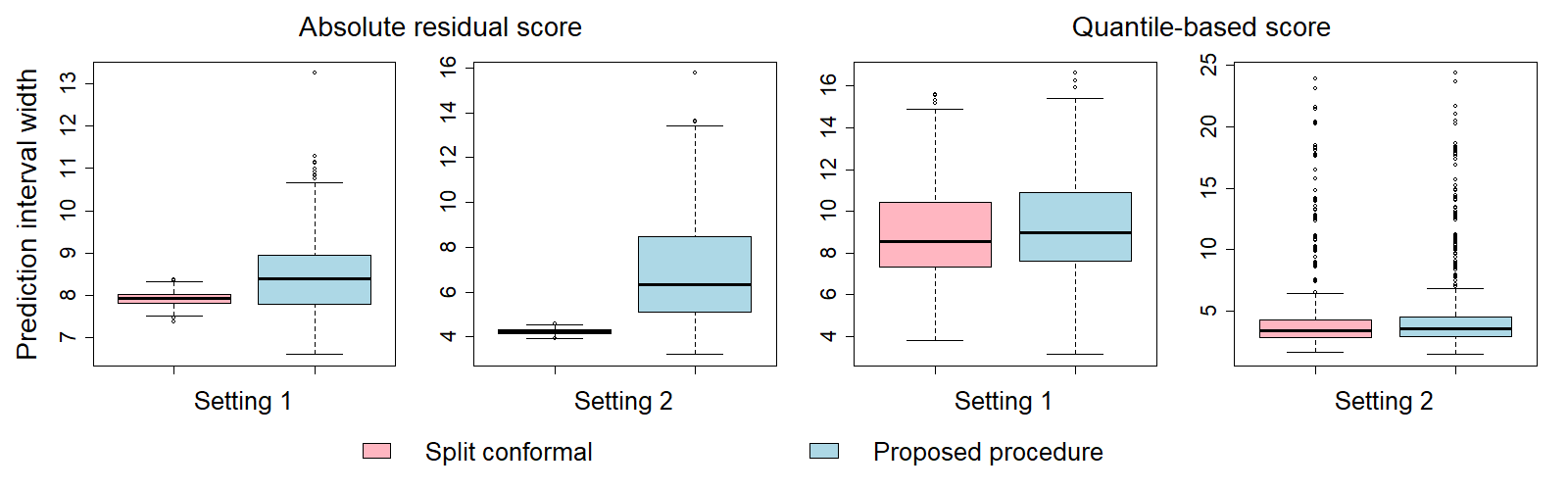}
    \caption{Widths of the prediction sets constructed using split conformal prediction and the proposed procedure with the absolute residual score and the quantile-based score, under Settings 1 and 2.}
    \label{fig:sim_width}
\end{figure}

In the case where the absolute residual score with linear regression is used (Figure~\ref{fig:sim_residual}), both split conformal prediction and the proposed method provide conditional miscoverage rates concentrated around $\alpha=0.2$ in Setting 1. However, in Setting 2, split conformal prediction fails to control the conditional miscoverage rate for a subset of points---likely those with $X > 8$. Recall that the residual score leads to a prediction set with homogeneous width across the feature space. The marginal coverage guarantee can still be satisfied even if conditional coverage is compromised for this subset, and split conformal prediction adopts this tradeoff, producing relatively narrower prediction sets. In contrast, the proposed method based on $L^2$-coverage yields smaller conditional miscoverage rates across most values of $X$ by producing wider prediction sets. 

In the experiments with the ``good score" (Figure~\ref{fig:sim_quantile}), what happens is quite different---the two methods produce similar prediction sets. The split conformal prediction no longer sacrifices conditional coverage quality for feature values in the ``hard" region $X > 8$ and achieves uniformly small conditional miscoverage rates. In this case, our proposed procedure is not forced to do much correction to achieve conditional coverage control and thus yields prediction sets similar to those of split conformal prediction. These results support that the proposed procedure does not become unnecessarily conservative, while tending to prioritize conditional miscoverage control in the tradeoff between the conditional miscoverage control and the conservativeness/prediction interval width.

In practice, when the data lie in a high-dimensional space and has a complex distributional structure, it is often challenging to determine in advance an appropriate choice of nonconformity score or evaluate the score we have. The selection of the method can thus be guided by the inferential objective: for instance, whether one prioritizes achieving better conditional miscoverage control (regardless of the quality of the score function), or whether one is satisfied with marginal coverage guarantees and prefers shorter prediction sets.

\begin{remark}
    In this section, we provided a comparison with standard conformal prediction rather than with other methods that aim for relaxed targets of conditional inference. The purpose of the experiments is to demonstrate that the proposed method performs as intended, achieving the target guarantee tightly without producing unnecessarily conservative prediction sets.

We note that the main advantage of the proposed method, relative to existing approaches on conditional predictive inference, lies in its interpretability. A meaningful empirical comparison with such methods is not straightforward, as they pursue different targets and naturally make different choices in the coverage–width tradeoff. Note that all methods output prediction sets of the form $\ch(x) = \{y : s(x,y) \leq \hat{q}\}$, where different methods yield different score bounds $\hat{q}$. Thus, for any method of this form with a score bound that is more conservative than another method (including those without theoretical guarantees), it is possible to claim based on a few plots that it provides stronger conditional coverage control. However, we do not consider this to be a particularly meaningful way of comparison, and instead view the discussions in Section~\ref{sec:multiaccuracy} as providing a more informative perspective.

In other words, the literature on distribution-free conditional predictive inference is generally not focused on ``improving a state-of-the-art methodology”---since ``improvement” is not well defined in the absence of a shared target---, and more about identifying what can and cannot be achieved in a distribution-free framework, as well as developing meaningful and interpretable conditions that are attainable.
\end{remark}

\subsection{Empirical evaluation of the proposed procedure}

We generate the data as follows:
\begin{align*}
    X \sim \mathcal{TN}(\mu,\Sigma;[0,5]^p),\quad Y \mid X \sim N\big(\beta_1^\top X + (\beta_2^\top X)^2 + \log |\beta_3^\top X|, (1+\exp(\beta_4^\top X))^2\big),
\end{align*}
where the feature dimension is set to $p = 10$, and the parameter vectors $\beta_1, \beta_2, \beta_3, \beta_4 \in \mathbb{R}^p$ are fixed in advance. Here, $\mathcal{TN}(\mu,\Sigma; A)$ denotes the truncated normal distribution with mean $\mu$, covariance matrix $\Sigma$, and truncation set $A$. In our experiment, we set $\mu = 2\mathbf{1}_p$ and $\Sigma = \mathbf{1}_p\mathbf{1}_p^\top + 2 I_p$, where $\mathbf{1}_p$ and $I_p$ denote the $p$-dimensional all-ones vector and the $p \times p$ identity matrix, respectively.

We generate a training dataset of size $500$, and fit random forest regression to construct an estimator $\hat{\mu}(\cdot)$. The nonconformity score function is then defined as $s(x, y) = |y - \hat{\mu}(x)|$. The function space used in the procedure is defined as the set of Gaussian kernels:
\begin{equation}\label{eqn:kernel_normal}
    \mathcal{F} = \{f_z : z \in \mathcal{X}\}, \quad \text{where } f_z(x) = \exp\left(-\frac{1}{2}\|(z - x)/h\|^2\right),
\end{equation}
and the experiments are repeated with varying bandwidths $h = 5, 10$, and $15$. We sample the functions as $f_{Z_1}, \cdots, f_{Z_m}$, where $Z_1,\cdots,Z_m$ are drawn i.i.d. from one of the following two distributions:
\[(1)\, Z \sim \mathcal{TN}(\mu,\Sigma; A), \qquad\qquad (2)\, Z \sim \text{Unif}([0,5]^p).\]
The first distribution reflects a setting in which unlabeled feature observations are used to construct the function samples, while the second distribution corresponds to a setting in which the function samples are constructed from a prespecified function. The latter can also be viewed as a setting in which feature observations are used under covariate shift.

We repeat the following steps 500 times: generate a calibration set of size $n = 2000$ and a test point $X_{n+1}$, apply the procedure described in Algorithm~\ref{alg:main} to construct a prediction set at levels $\alpha=0.05,0.1,0.15,0.2,0.25,0.3$, and compute the conditional miscoverage rates $\alpha_\D(X_{n+1})$ and $\ba_\D(f)$. For the computation of $\ba_\D(f)$, we generate 50,000 samples of $(X,Y)$ to estimate the quantity numerically. The results under function-sampling scheme (1) and (2) are presented in Figures~\ref{fig:sim_mult} and~\ref{fig:hist_mult}, respectively.

\begin{figure}[htbp]
    \centering
    \includegraphics[width=0.9\linewidth]{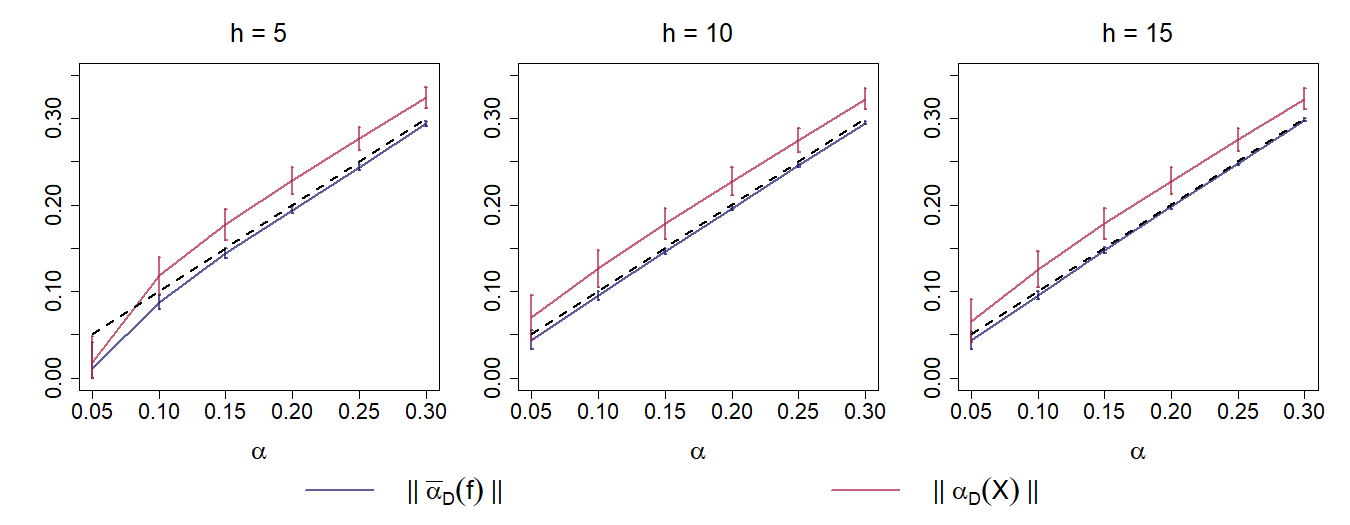}
    \caption{Plots of $\|\ba_\D(f)\|$ and $\|\alpha_\D(X)\|$ from the proposed method using Gaussian kernels with varying bandwidths and target levels, under function-sampling scheme (1). The dotted line corresponds to the line $y=x$.}
    \label{fig:sim_mult}
\end{figure}

\begin{figure}[htbp]
    \centering
    \includegraphics[width=0.85\linewidth]{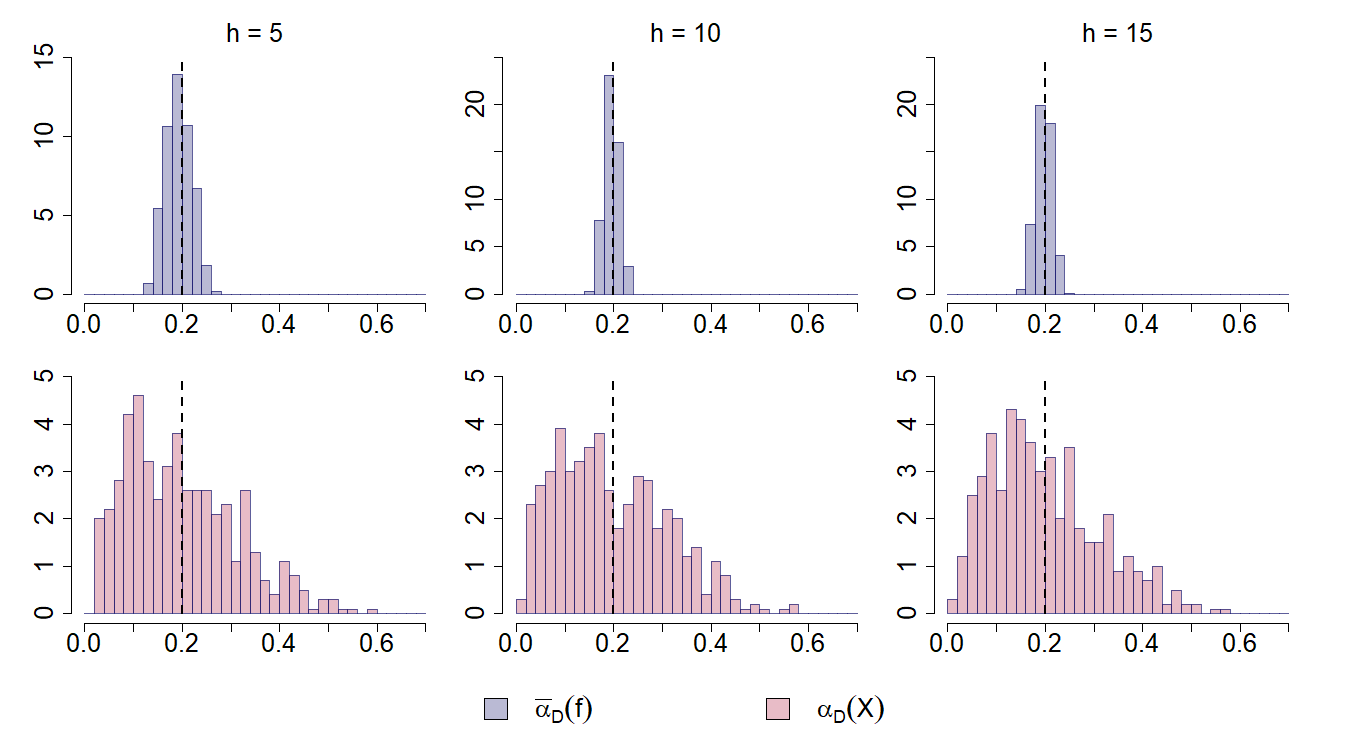}
    \caption{Histograms of $\ba_\D(f)$ and $\alpha_\D(X)$ from the proposed method with bandwidth $h=5, 10$, and $15$, under function-sampling scheme (1) with $\alpha = 0.2$. The dotted line corresponds to the line $x = 0.2$.}
    \label{fig:hist_mult}
\end{figure}

Figure~\ref{fig:sim_mult} illustrates that the proposed method tightly controls $\|\ba_\D(f)\|_2$ at the target level $\alpha$, while $\|\alpha_\D(X)\|_2$ is also approximately controlled across different bandwidth choices. Recall that $\ba_\D(f)$ can be interpreted as a kernel-weighted average of $\alpha_\D(X)$; by controlling these weighted averages, the method effectively keeps the conditional miscoverage rates small. This is also illustrated in Figure~\ref{fig:hist_mult}. The $L^2$-coverage condition causes $\ba_\D(f)$ to be distributed around the target level $\alpha$, which in turn leads to $\alpha_\D(X)$ also concentrating around $\alpha$ (to a lesser extent).

Next, Figures~\ref{fig:sim_mult_2} and~\ref{fig:hist_mult_2} present the results under sampling scheme (2), showing similar patterns to those observed under scheme (1). Note that in Figure~\ref{fig:sim_mult_2}, the $L^2$-norms $\|\ba_\D(f)\|$ and $\|\alpha_\D(X)\|$ are now computed with respect to $Z \sim \text{Unif}([0,5]^p)$ and $X \sim \text{Unif}([0,5]^p)$, respectively.

\begin{figure}[htbp]
    \centering
    \includegraphics[width=0.9\linewidth]{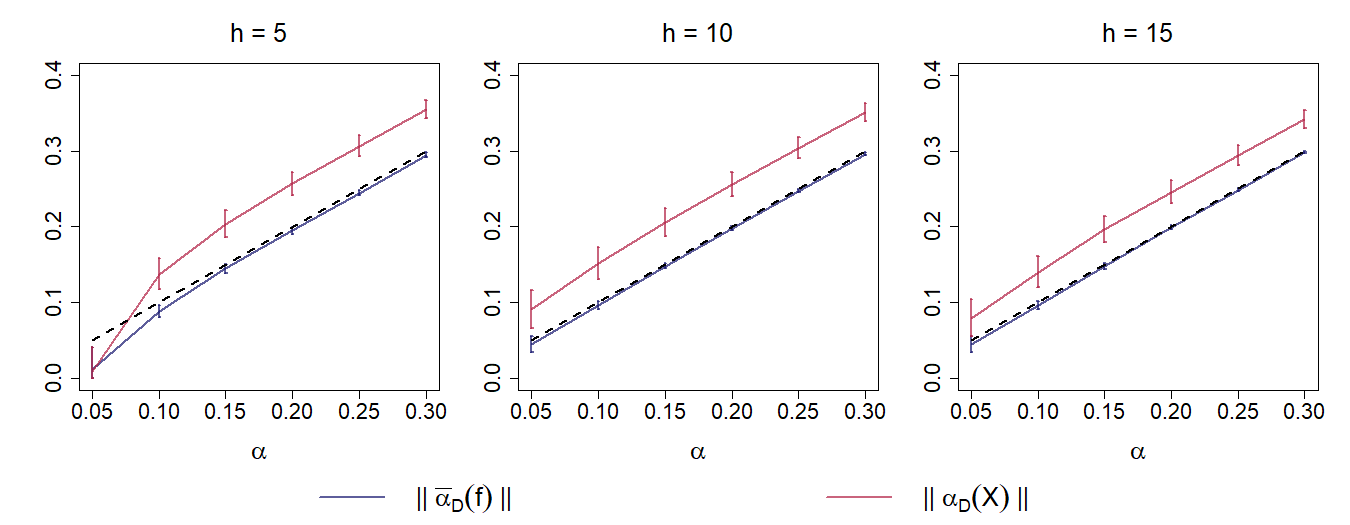}
    \caption{Plots of $\|\ba_\D(f)\|$ and $\|\alpha_\D(X)\|$ from the proposed method using Gaussian kernels with varying bandwidths and target levels, under function-sampling scheme (2). The dotted line corresponds to the line $y=x$.}
    \label{fig:sim_mult_2}
\end{figure}

\begin{figure}[htbp]
    \centering
    \includegraphics[width=0.85\linewidth]{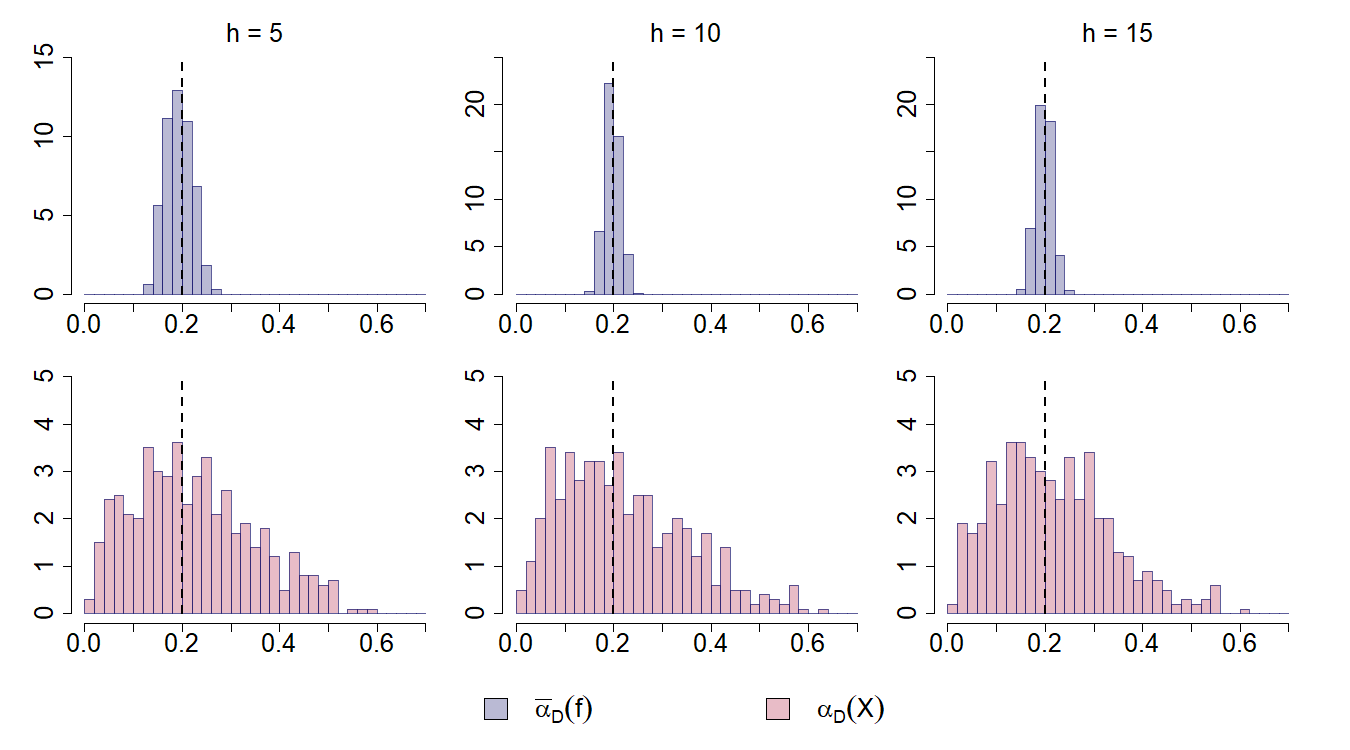}
    \caption{Histograms of $\ba_\D(f)$ and $\alpha_\D(X)$ from the proposed method with bandwidth $h=5, 10$, and $15$, under function-sampling scheme (2) with $\alpha = 0.2$. The dotted line corresponds to the line $x = 0.2$.}
    \label{fig:hist_mult_2}
\end{figure}

\section{Illustration on Abalone dataset}

We further illustrate the performance of our procedure by applying it to the Abalone dataset~\citep{nash1994population}, which was also previously applied in~\citet{hore2023conformal}. The dataset consists of 4,177 samples of abalones, along with feature information on sex and physical measurements such as length, height, and weight. The outcome variable is the number of growth rings, which represents the age of the abalones.

In the experiment, we include seven features: sex, length, diameter, height, shucked weight, viscera weight, and shell weight. We split the data into four subsets: training data $\dt$ of size 676, calibration data $\dc$ of size 2000, data used to construct the function samples $\df$ and $f$ of total size 1001, and test data $\D_\text{test}$ of size 500. As in the simulations, we use Gaussian kernels~\eqref{eqn:kernel_normal} to construct the 1001 function samples, with centers chosen from the data---we apply some rescaling to the data in advance so that distances across different features are reflected uniformly in the kernel function values. We apply random forest regression to construct an estimator $\hat{\mu}(\cdot)$ and the score function $|y - \hat{\mu}(x)|$. Using the calibration data and the function samples, we construct prediction sets at the 500 test points following Procedure~\ref{alg:main}.

To illustrate the performance of the procedure in terms of conditional coverage rates---which cannot be obtained from the realized data---we instead examine the following quantities:

\begin{enumerate}
    \item Estimate of $\ba_\D(f_{\tilde{x}})$: 
    \[\hat{\alpha}(f_{\tilde{x}}) = \frac{1}{|\D_\text{test}|}\sum_{ (x,y) \in \D_\text{test}} \tilde{f}_{\tilde{x}}(x) \One{y \notin \ch(x)},\]
    \item Estimate of local-conditional miscoverage rate: 
    \[\hat{\alpha}_\text{local}(\tilde{x}) = \frac{1}{N_{\tilde{x}}} \sum_{(x,y) \in \D_\text{test} : \|x - \tilde{x}\| \leq 1} \One{y \notin \ch(x)}, \text{ where } N_{\tilde{x}} = |\{x,y) \in \D_\text{test} : \|x - \tilde{x}\| \leq 1\}|,\]
\end{enumerate}
where the above quantities are computed for each feature observation $\tilde{x}$ in the test data $\D_\text{test}$.

Figure~\ref{fig:abalone} presents the plot of $|\ba_\D(f)|$, estimated using the values of $\hat{\alpha}(f_{\tilde{x}})$, under varying kernel bandwidths and levels. Although the estimation may not be entirely accurate due to the availability of only a single realization of the data---effectively making this a conditional $L^2$-coverage estimate---we include this plot to provide a rough illustration of the overall performance of the proposed procedure. Figure~\ref{fig:abalone} illustrates that the proposed procedure successfully controls $|\ba_\D(f)|$ at the desired level $\alpha$ across various settings. In contrast, split conformal prediction does not directly control $|\ba_\D(f)|$, but it approximately does so when the bandwidth of the kernel functions is large. This behavior is expected, as $\ba_\D(f)$ converges to the marginal miscoverage rate as the bandwidth increases, which is controlled to be less than or equal to $\alpha$ by the split conformal prediction method. 

\begin{figure}[htbp]
    \centering
    \includegraphics[width=0.85\linewidth]{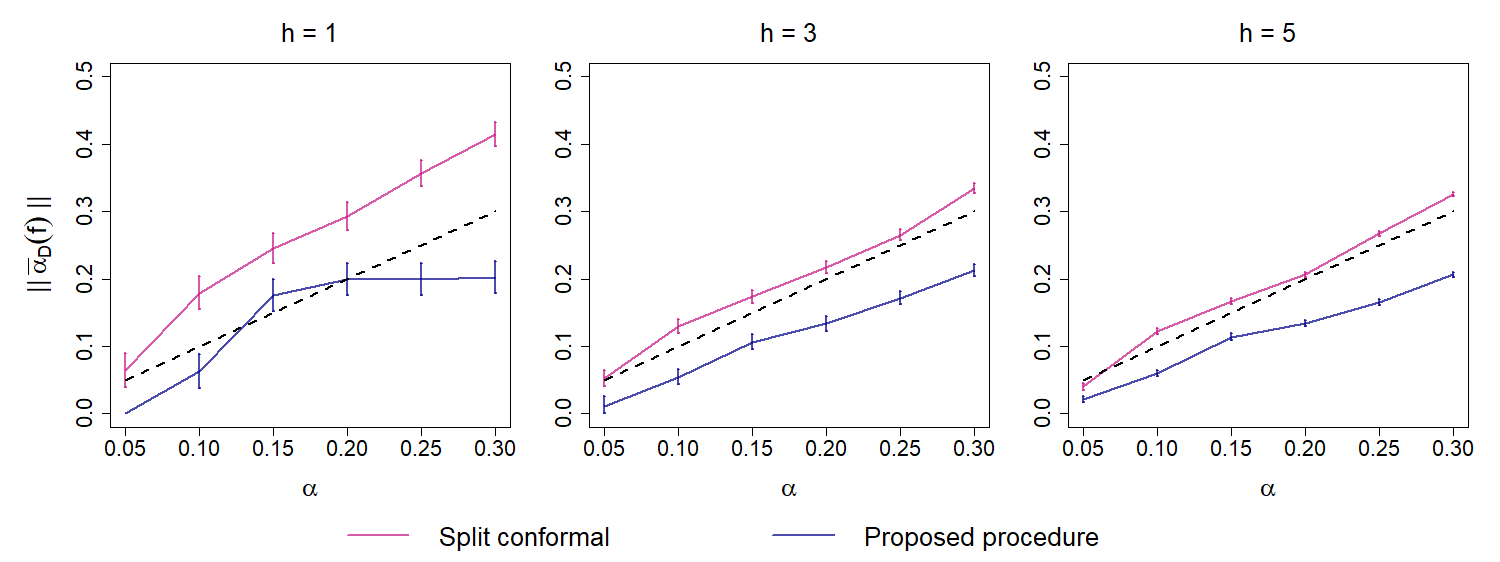}
    \caption{Results for the Abalone dataset: plots of $\|\ba_\D(f)\|_2$ from the split conformal prediction and the proposed procedure, under different kernel bandwidths and values of $\alpha$.
    }
    \label{fig:abalone}
\end{figure}

Figure~\ref{fig:abalone_hist} shows the histograms of the estimated $\ba(f)$ values from the two methods, illustrating that the proposed procedure provides stronger control over these values---the proposed procedure attains $\ba(f) \leq \alpha$ in almost all trials, while the standard conformal prediction does not. The corresponding results for the local-conditional miscoverage rate are shown in Figure~\ref{fig:abalone_local}, again illustrating stronger control by the proposed procedure. 

\begin{figure}[htbp]
    \centering
    \includegraphics[width=0.8\linewidth]{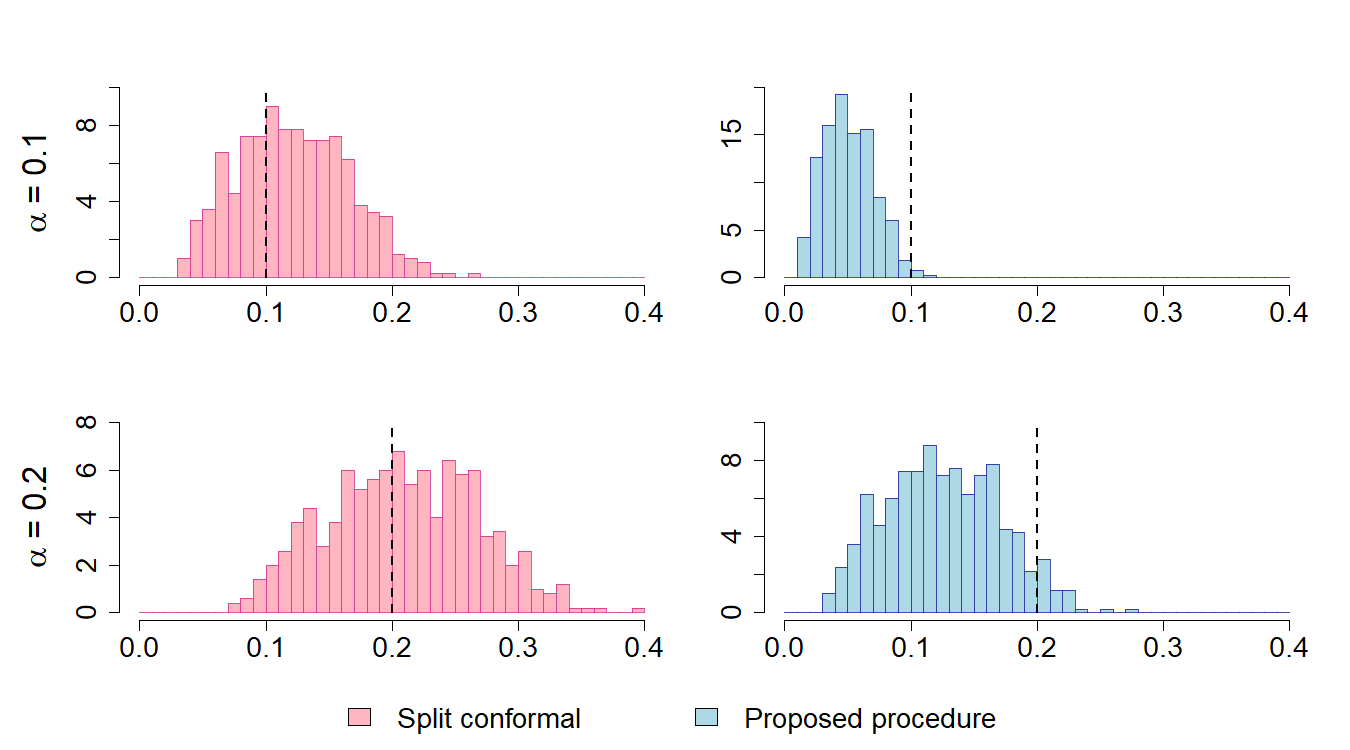}
    \caption{Results for the Abalone dataset: histograms of $\ba(f)$ from the split conformal prediction and the proposed procedure, under bandwidth $h=3$ and $\alpha=0.1, 0.2$.}
    \label{fig:abalone_hist}
\end{figure}

\begin{figure}[htbp]
    \centering
    \includegraphics[width=0.8\linewidth]{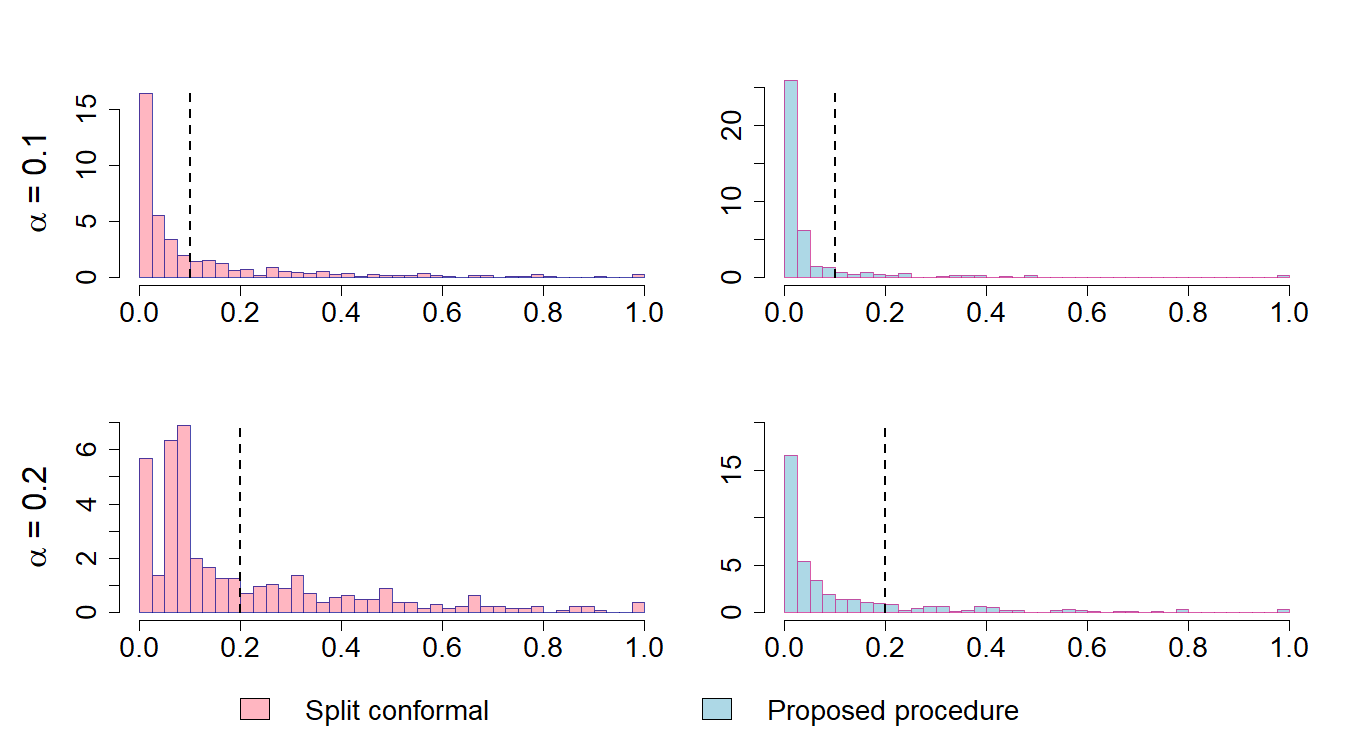}
    \caption{Results for the Abalone dataset: histograms of the local-conditional miscoverage rate for the split conformal prediction and the proposed procedure, under bandwidth $h=3$ and $\alpha = 0.1, 0.2$.}
    \label{fig:abalone_local}
\end{figure}

\section{Discussion}

In this work, we present a formulation for distribution-free conditional predictive inference that is both simple and broadly applicable. We introduce an inferential target—the $L^k$-coverage guarantee—which appears to be a near-minimal relaxation of the strict conditional coverage requirement that is otherwise unachievable, and we develop a procedure that attains this guarantee with exact finite-sample validity under no distributional assumptions. Empirical results support that our procedure provides effective control of conditional coverage rates, while avoiding to be unnecessarily conservative.

The ideas introduced in this work---including viewing conditional coverage as a function or random variable, and employing artificial hierarchical structures through sample grouping---may have broader applications to inference problems across various contexts. Extensions to non-i.i.d. data or structured data may be of potential interest, as well as inference on multiple outcomes with conditional guarantees. We leave these questions to future work.

\subsection*{Acknowledgment}
Z.R.~is supported by NSF grant DMS-2413135 and Wharton Analytics. Y.L.~is supported in part by 
NIH R01-AG065276, R01-GM139926, NSF 2210662, P01-AG041710, R01-CA222147, 
ARO W911NF-23-1-0296,
NSF 2046874, ONR N00014-21-1-2843, and the Sloan Foundation.

\bibliographystyle{plainnat}
\bibliography{bib}

\begin{thebibliography}{25}
\providecommand{\natexlab}[1]{#1}
\providecommand{\url}[1]{\texttt{#1}}
\expandafter\ifx\csname urlstyle\endcsname\relax
  \providecommand{\doi}[1]{doi: #1}\else
  \providecommand{\doi}{doi: \begingroup \urlstyle{rm}\Url}\fi

\bibitem[Angelopoulos and Bates(2021)]{angelopoulos2021gentle}
Anastasios~N Angelopoulos and Stephen Bates.
\newblock A gentle introduction to conformal prediction and distribution-free uncertainty quantification.
\newblock \emph{arXiv preprint arXiv:2107.07511}, 2021.

\bibitem[Angelopoulos et~al.(2024)Angelopoulos, Barber, and Bates]{angelopoulos2024theoretical}
Anastasios~N Angelopoulos, Rina~Foygel Barber, and Stephen Bates.
\newblock Theoretical foundations of conformal prediction.
\newblock \emph{arXiv preprint arXiv:2411.11824}, 2024.

\bibitem[Barber(2020)]{barber2020distribution}
Rina~Foygel Barber.
\newblock Is distribution-free inference possible for binary regression?
\newblock \emph{Electronic Journal of Statistics}, 14\penalty0 (2), 2020.

\bibitem[Barber et~al.(2021)Barber, Candes, Ramdas, and Tibshirani]{foygel2021limits}
Rina~Foygel Barber, Emmanuel~J Candes, Aaditya Ramdas, and Ryan~J Tibshirani.
\newblock The limits of distribution-free conditional predictive inference.
\newblock \emph{Information and Inference: A Journal of the IMA}, 10\penalty0 (2):\penalty0 455--482, 2021.

\bibitem[Bates et~al.(2021)Bates, Angelopoulos, Lei, Malik, and Jordan]{bates2021distribution}
Stephen Bates, Anastasios Angelopoulos, Lihua Lei, Jitendra Malik, and Michael Jordan.
\newblock Distribution-free, risk-controlling prediction sets.
\newblock \emph{Journal of the ACM (JACM)}, 68\penalty0 (6):\penalty0 1--34, 2021.

\bibitem[Bian and Barber(2023)]{bian2023training}
Michael Bian and Rina~Foygel Barber.
\newblock Training-conditional coverage for distribution-free predictive inference.
\newblock \emph{Electronic Journal of Statistics}, 17\penalty0 (2):\penalty0 2044--2066, 2023.

\bibitem[Chernozhukov et~al.(2021)Chernozhukov, W{\"u}thrich, and Zhu]{chernozhukov2021distributional}
Victor Chernozhukov, Kaspar W{\"u}thrich, and Yinchu Zhu.
\newblock Distributional conformal prediction.
\newblock \emph{Proceedings of the National Academy of Sciences}, 118\penalty0 (48):\penalty0 e2107794118, 2021.

\bibitem[Gibbs and Cand{\`e}s(2025)]{gibbs2025characterizing}
Isaac Gibbs and Emmanuel~J Cand{\`e}s.
\newblock Characterizing the training-conditional coverage of full conformal inference in high dimensions.
\newblock \emph{arXiv preprint arXiv:2502.20579}, 2025.

\bibitem[Gibbs et~al.(2025)Gibbs, Cherian, and Cand{\`e}s]{gibbs2025conformal}
Isaac Gibbs, John~J Cherian, and Emmanuel~J Cand{\`e}s.
\newblock Conformal prediction with conditional guarantees.
\newblock \emph{Journal of the Royal Statistical Society Series B: Statistical Methodology}, page qkaf008, 2025.

\bibitem[Guan(2023)]{guan2023localized}
Leying Guan.
\newblock Localized conformal prediction: A generalized inference framework for conformal prediction.
\newblock \emph{Biometrika}, 110\penalty0 (1):\penalty0 33--50, 2023.

\bibitem[Hebert-Johnson et~al.(2018)Hebert-Johnson, Kim, Reingold, and Rothblum]{pmlr-v80-hebert-johnson18a}
Ursula Hebert-Johnson, Michael Kim, Omer Reingold, and Guy Rothblum.
\newblock Multicalibration: Calibration for the ({C}omputationally-identifiable) masses.
\newblock In Jennifer Dy and Andreas Krause, editors, \emph{Proceedings of the 35th International Conference on Machine Learning}, volume~80 of \emph{Proceedings of Machine Learning Research}, pages 1939--1948. PMLR, 10--15 Jul 2018.
\newblock URL \url{https://proceedings.mlr.press/v80/hebert-johnson18a.html}.

\bibitem[Hore and Barber(2023)]{hore2023conformal}
Rohan Hore and Rina~Foygel Barber.
\newblock Conformal prediction with local weights: randomization enables local guarantees.
\newblock \emph{arXiv preprint arXiv:2310.07850}, 2023.

\bibitem[Kim et~al.(2019)Kim, Ghorbani, and Zou]{kim2019multiaccuracy}
Michael~P Kim, Amirata Ghorbani, and James Zou.
\newblock Multiaccuracy: Black-box post-processing for fairness in classification.
\newblock In \emph{Proceedings of the 2019 AAAI/ACM Conference on AI, Ethics, and Society}, pages 247--254, 2019.

\bibitem[Lee and Barber(2021)]{lee2021distribution}
Yonghoon Lee and Rina Barber.
\newblock Distribution-free inference for regression: discrete, continuous, and in between.
\newblock \emph{Advances in Neural Information Processing Systems}, 34:\penalty0 7448--7459, 2021.

\bibitem[Lee et~al.(2023)Lee, Barber, and Willett]{lee2023distribution}
Yonghoon Lee, Rina~Foygel Barber, and Rebecca Willett.
\newblock Distribution-free inference with hierarchical data.
\newblock \emph{arXiv preprint arXiv:2306.06342}, 2023.

\bibitem[Liang and Barber(2023)]{liang2023algorithmic}
Ruiting Liang and Rina~Foygel Barber.
\newblock Algorithmic stability implies training-conditional coverage for distribution-free prediction methods.
\newblock \emph{arXiv preprint arXiv:2311.04295}, 2023.

\bibitem[Medarametla and Cand{\`e}s(2021)]{medarametla2021distribution}
Dhruv Medarametla and Emmanuel Cand{\`e}s.
\newblock Distribution-free conditional median inference.
\newblock \emph{Electronic Journal of Statistics}, 15\penalty0 (2):\penalty0 4625--4658, 2021.

\bibitem[Nash et~al.(1994)Nash, Sellers, Talbot, Cawthorn, and Ford]{nash1994population}
Warwick~J Nash, Tracy~L Sellers, Simon~R Talbot, Andrew~J Cawthorn, and Wes~B Ford.
\newblock The population biology of abalone (haliotis species) in tasmania. i. blacklip abalone (h. rubra) from the north coast and islands of bass strait.
\newblock \emph{Sea Fisheries Division, Technical Report}, 48:\penalty0 p411, 1994.

\bibitem[Papadopoulos et~al.(2002)Papadopoulos, Proedrou, Vovk, and Gammerman]{papadopoulos2002inductive}
Harris Papadopoulos, Kostas Proedrou, Volodya Vovk, and Alex Gammerman.
\newblock Inductive confidence machines for regression.
\newblock In \emph{Machine learning: ECML 2002: 13th European conference on machine learning Helsinki, Finland, August 19--23, 2002 proceedings 13}, pages 345--356. Springer, 2002.

\bibitem[Romano et~al.(2019)Romano, Patterson, and Candes]{romano2019conformalized}
Yaniv Romano, Evan Patterson, and Emmanuel Candes.
\newblock Conformalized quantile regression.
\newblock \emph{Advances in neural information processing systems}, 32, 2019.

\bibitem[Shafer and Vovk(2008)]{shafer2008tutorial}
Glenn Shafer and Vladimir Vovk.
\newblock A tutorial on conformal prediction.
\newblock \emph{Journal of Machine Learning Research}, 9\penalty0 (3), 2008.

\bibitem[Vovk(2012)]{vovk2012conditional}
Vladimir Vovk.
\newblock Conditional validity of inductive conformal predictors.
\newblock In \emph{Asian conference on machine learning}, pages 475--490. PMLR, 2012.

\bibitem[Vovk et~al.(2005)Vovk, Gammerman, and Shafer]{vovk2005algorithmic}
Vladimir Vovk, Alexander Gammerman, and Glenn Shafer.
\newblock \emph{Algorithmic learning in a random world}, volume~29.
\newblock Springer, 2005.

\bibitem[Xie et~al.(2024)Xie, Barber, and Candes]{xie2024boosted}
Ran Xie, Rina Barber, and Emmanuel Candes.
\newblock Boosted conformal prediction intervals.
\newblock \emph{Advances in Neural Information Processing Systems}, 37:\penalty0 71868--71899, 2024.

\bibitem[Zhang and Cand{\`e}s(2024)]{zhang2024posterior}
Yao Zhang and Emmanuel~J Cand{\`e}s.
\newblock Posterior conformal prediction.
\newblock \emph{arXiv preprint arXiv:2409.19712}, 2024.

\end{thebibliography}

\newpage
\appendix

\section{Equivalence of the coverage conditions}\label{sec:multi_acc}
Here, we provide the proof for the equivalence between conditions (i) and (ii) in Equation~\eqref{eqn:multi_acc_ineq}.

\begin{proposition}\label{prop:equiv}
    The conditions (i) and (ii) in Equation~\eqref{eqn:multi_acc_ineq} are equivalent.
\end{proposition}

\begin{proof}[Proof of Proposition~\ref{prop:equiv}]
    It is trivial that the condition (i) implies the condition (ii). Now suppose that (i) does not hold, meaning that $\delta := \PP{\alpha(X_{n+1}) > \alpha} > 0$. 
    Now define the events
    \[E_k = \left\{\alpha(X_{n+1}) \geq \alpha + \frac{1}{k}\right\}, \qquad \text{ for } k = 1,2,\cdots.\]
Since $(E_k)_{k=1,2 \cdots}$ is an increasing sequence of events and $\cup_{k=1}^\infty E_k = \{\alpha(X_{n+1}) > \alpha\}$, we have $\lim_{k \rightarrow \infty} \PP{E_k} = \delta > 0$, implying that there exists a $K$ such that $\PP{E_K} > 0$. Letting $f_K = \mathbbm{1}_{\{x : \alpha(x) \geq \alpha + 1/K\}}$, we have
\begin{multline*}
\EE{f_K(X_{n+1})\left(\alpha(X_{n+1}) - \alpha\right)} = \EE{\One{\alpha(X_{n+1}) \geq \alpha + \tfrac{1}{K}}\cdot\left(\alpha(X_{n+1}) - \alpha\right)}\\
\geq \EE{\One{\alpha(X_{n+1}) \geq \alpha + \tfrac{1}{K}}\cdot \tfrac{1}{K}} = \tfrac{1}{K}\cdot\PP{E_K} > 0,
\end{multline*}
contradicting condition (ii). This proves that (ii) implies (i), and the claim is proved.

\end{proof}

\section{Procedure with \texorpdfstring{$L^k$}{}-coverage guarantee}\label{sec:method_k}

Here, we provide the details of the method with $L^k$-coverage guarantee, which can be obtained through a direct extension of procedure~\ref{alg:main} with the $L^2$-coverage guarantee. 
Suppose we have datasets $\dt$, $\dc$, and the function space $\F$ with a distribution $P_f$ and a uniform bound $b$. We first construct a nonconformity score function $s : \X \times \Y \rightarrow \R^+$, using the training data $\dt$. Then let $n_\text{train} = m_\text{train}k + r_\text{train}$, where $0 \leq r_\text{train} \leq k-1$. For each $f \in \F$, we define
\[\gamma(f) = \left(\frac{1}{\mt+1}\left( \sum_{l=1}^{\mt} \prod_{j=1}^k f(X_{(l-1)k+j}') + b^k\right)\right)^{1/k}.\]
Next, we write $n = mk+r$, where $0 \leq r \leq k-1$, and draw function samples $f_1,\cdots,f_m \iidsim P_f$ and $f \sim P_f$. Then we construct the prediction set as
\begin{equation}\label{eqn:chat_k}
\begin{split}
    &\ch(x) = \{y \in \Y : s(x,y) \leq \hat{t}\}, \text{ where }\\
    &\hat{t} = \min\left\{t \geq 0 : \frac{\sum_{i=1}^m  \prod_{j=1}^k \tilde{f}_i(X_{(i-1)k+j}) Z_{(i-1)k+j}^t+\frac{b^k}{\gamma(f)^k}}{\sum_{i=1}^m \prod_{j=1}^k \tilde{f}_i(X_{(i-1)k+j})} \leq \alpha^k\right\},
\end{split}
\end{equation}
where $\tilde{f}_i = f_i / \gamma(f)$ for $i \in [m]$.

\begin{theorem}\label{thm:cov_k}
    The prediction set $\ch(X_{n+1})$ given as~\eqref{eqn:chat_k} satisfies
    \[\|\ta_\D(f)\|_k \leq \alpha, \text{ where } \ta_\D(f) = \EEst{\frac{f(X_{n+1})}{\gamma(f)} \cdot \alpha_\D(X_{n+1})}{\D}.\]
\end{theorem}
We omit the proof of Theorem~\ref{thm:cov_k} since it follows from applying the same arguments from Theorem~\ref{thm:main}. While the above procedure provides a general recipe for achieving the $L^k$-coverage guarantee which provides a theoretically stronger control over the conditional miscoverage rates for larger $k$, we present and recommend procedure~\ref{alg:main} with the $L^2$-coverage guarantee as the primary methodology for the following reasons:
\begin{enumerate}
    \item The prediction set~\eqref{eqn:chat_k} tends to be conservative unless the sample size is substantially large. To see that, observe that in the definition of $\hat{t}$, we are essentially comparing a ratio of the form $\sum_{i=1}^m w_i \cdot (\text{product of $Z_i$'s}) / \sum_{i=1}^m w_i$ to $\alpha^k$---roughly speaking, we are comparing a ratio of order $1/m$ with $\alpha^k$, and unless $m$ is very large, it leads to a trivial prediction set.

    \item The $L^2$-coverage guarantee already provides good control of the conditional coverage rates. This is supported by the arguments presented in Section~\ref{sec:moment} and the empirical results in Section~\ref{sec:sim}.
\end{enumerate}

\section{Proof of Theorems}

\subsection{Proof of Theorem~\ref{thm:main}}
\label{appd:proof_main}
Let $(X_{n+2},Y_{n+2})$ be another (hypothetical) sample from $P_{X,Y}$, drawn independently of the dataset, and then define $Z_{n+j}^t = \One{Y_{n+j} \notin C_t(X_{n+j})}$ for $j=1,2$ and $t \geq 0$. Recall that $\D = (\dt, \dc, \df, f)$.
We have
\begin{align*}
    &\ta_\D(f)^2 = \EEst{\frac{f(X_{n+1})}{\gamma(f)} \cdot \alpha_\D(X_{n+1})}{\D}^2\\
    &= \frac{1}{\gamma(f)^2}\cdot \EEst{f(X_{n+1}) \cdot \alpha_\D(X_{n+1})}{\D}^2\hspace{25mm}\textnormal{ since $\gamma$ is a function of $f$, $\dt$, and $\df$}\\
    &= \frac{1}{\gamma(f)^2}\cdot \EEst{f(X_{n+1}) \cdot \alpha_\D(X_{n+1})}{\D}\cdot\EEst{f(X_{n+2}) \cdot \alpha_\D(X_{n+2})}{\D}\\
    &= \frac{1}{\gamma(f)^2}\cdot \EEst{f(X_{n+1}) \cdot  Z_{n+1}^{\hat{t}}}{\D}\cdot\EEst{f(X_{n+2}) \cdot  Z_{n+2}^{\hat{t}}}{\D}\\
    &\hspace{38mm}\textnormal{by the definition of $\alpha_\D$ and the fact that $\hat{t}$ is a function of $\D$}\\
    &= \frac{1}{\gamma(f)^2}\cdot \EEst{f(X_{n+1})f(X_{n+2}) \cdot Z_{n+1}^{\hat{t}}Z_{n+2}^{\hat{t}}}{\D}\\
    &\hspace{50mm}\textnormal{ since $(X_{n+1},Y_{n+1}), (X_{n+2},Y_{n+2})$, and $\D$ are mutually independent}\\
    &= \EEst{\tilde{f}(X_{n+1})\tilde{f}(X_{n+2}) \cdot Z_{n+1}^{\hat{t}}Z_{n+2}^{\hat{t}}}{\D},
\end{align*}
where $\tilde{f} = f/\gamma(f)$. Next, define $\tilde{t}$ by
\[\tilde{t} = \min\left\{t \geq 0 : \frac{\sum_{i=1}^{m+1}  \tilde{f}_i(X_{2i-1})\tilde{f}_i(X_{2i})Z_{2i-1}^tZ_{2i}^t}{\sum_{i=1}^{m+1} \tilde{f}_i(X_{2i-1})\tilde{f}_i(X_{2i})} \leq \alpha^2\right\},\]
where, for convenience, we write $\tilde{f}_{m+1} = \tilde{f}$. The value $\tilde{t}$ can be viewed as an `oracle' cutoff that depends on $(X_{2m+1},Y_{2m+1})$ and $(X_{2m+2},Y_{2m+2})$, which are not observed. Note that by definition, $\hat{t} \geq \tilde{t}$, and thus by the monotonicity of $C_t$---i.e., $C_t(x) \subset C_{t'}(x)$ if $t \leq t'$---, it follows that $Z_i^{\hat{t}} \leq Z_i^{\tilde{t}}$ for all $i \in [n+2]$. Therefore, from the above result, we have
\begin{align*}
    &\EEst{\ta_\D(f)^2}{\dt} = \EEst{\EEst{\tilde{f}(X_{n+1})\tilde{f}(X_{n+2}) \cdot Z_{n+1}^{\hat{t}}Z_{n+2}^{\hat{t}}}{\D}}{\dt}\\
    &\leq \EEst{\EEst{\tilde{f}(X_{n+1})\tilde{f}(X_{n+2}) \cdot Z_{n+1}^{\tilde{t}}Z_{n+2}^{\tilde{t}}}{\D}}{\dt}\\
    &=\EEst{\tilde{f}(X_{n+1})\tilde{f}(X_{n+2}) \cdot Z_{n+1}^{\tilde{t}}Z_{n+2}^{\tilde{t}}}{\dt}\\
    &= \EEst{\frac{1}{m+1}\cdot \sum_{i=1}^{m+1}\tilde{f}_i(X_{2i-1})\tilde{f}_i(X_{2i}) \cdot Z_{n+1}^{\tilde{t}}Z_{n+2}^{\tilde{t}}}{\dt},
\end{align*}
where the last equality follows from the exchangeability of $f_1, \cdots, f_m, f_{m+1} = f$ and the invariance of $\tilde{t}$ under arbitrary permutations of $(f_i, X_{2i-1}, X_{2i})_{i \in [m]}$. From the definition of $\tilde{t}$, we then have---where we exclude the case where $\sum_{i=1}^{m+1} \tilde{f}_i(X_{2i-1})\tilde{f}_i(X_{2i}) = 0$ and consequently $\hat{t} = \tilde{t} = \infty$, in which case the guarantee holds trivially---
\begin{align*}
    \EEst{\ta_\D(f)^2}{\dt} \leq \alpha^2 \cdot \EEst{\frac{1}{m+1}\cdot \sum_{i=1}^{m+1} \tilde{f}_i(X_{2i-1})\tilde{f}_i(X_{2i})}{\dt} = \alpha^2\cdot \EEst{\tilde{f}(X_{n+1})\tilde{f}(X_{n+2})}{\dt},
\end{align*}
and thus, by marginalizing with respect to $\dt$, we obtain
\begin{multline}\label{eqn:pf_1}
    \EE{\ta_\D(f)^2} \leq \alpha^2 \cdot \EE{\tilde{f}(X_{n+1})\tilde{f}(X_{n+2})} 
    = \alpha ^2 \cdot \EE{f(X_{n+1})f(X_{n+2}) / \gamma(f)^2}\\
    = \alpha^2 \cdot \EE{\frac{(\mt+1)\cdot f(X_{n+1})f(X_{n+2})}{\sum_{l=1}^{\mt} f(X_{2l-1}') f(X_{2l}') + b^2}} \leq \alpha^2 \cdot \EE{\frac{(\mt+1)\cdot f(X_{n+1})f(X_{n+2})}{\sum_{l=1}^{\mt} f(X_{2l-1}') f(X_{2l}') + f(X_{n+1})f(X_{n+2})}} = \alpha^2,
\end{multline}
where the last equality follows from the exchangeability of $(X_1',X_2'), \dots, (X_{2\mt-1}, X_{2\mt}), (X_{n+1}, X_{n+2})$, which are i.i.d. samples from $P_X \times P_X$.

\subsection{Proof of Theorem~\ref{thm:generalized}}

Let us define $(X_{m+1,1},Y_{m+1,1}) = (X_{n+1}, Y_{n+1})$ and consider hypothetical samples $(X_{m+1,2},Y_{m+1,2}), \cdots$, $(X_{m+1,r},Y_{m+1,r})$ $\iidsim P_X \times P_{Y \mid X}$. Then we have
\begin{align*}
    \ta_\D(f)^2 &= \EEst{\tilde{f}(X_{m+1,1})\tilde{f}(X_{m+1,2}) \cdot Z_{m+1,1}^{\hat{t}}Z_{m+1,2}^{\hat{t}}}{\D} \qquad \text{ by the result in the proof of Theorem~\ref{thm:main}}\\
    &= \frac{1}{\binom{r}{2}}\sum_{1 \leq j_1 < j_2 \leq r} \EEst{\tilde{f}(X_{m+1,j_1})\tilde{f}(X_{m+1,j_2}) \cdot Z_{m+1,j_1}^{\hat{t}}Z_{m+1,j_2}^{\hat{t}}}{\D}\\
    &= \EEst{\frac{1}{\binom{r}{2}}\sum_{1 \leq j_1 < j_2 \leq r}\tilde{f}(X_{m+1,j_1})\tilde{f}(X_{m+1,j_2}) \cdot Z_{m+1,j_1}^{\hat{t}}Z_{m+1,j_2}^{\hat{t}}}{\D}
\end{align*}
where $Z_{m+1,j}^t = \One{s(X_{m+1,j},Y_{m+1,j}) > t}$ for $j \in [r]$, and $t \geq 0$. Now, define
\[\tilde{t} = \min\left\{t \geq 0 : \frac{\sum_{k=1}^{m+1} \frac{1}{\binom{r}{2}} \sum_{1 \leq j_1 < j_2 \leq r} \tilde{f}_k(X_{k,j_1})\tilde{f}_k(X_{k,j_2})Z_{k,j_1}^t Z_{k,j_2}^t}{\sum_{k=1}^{m+1} \frac{1}{\binom{r}{2}} \sum_{1 \leq j_1 < j_2 \leq r} \tilde{f}_k(X_{k,j_1})\tilde{f}_k(X_{k,j_2})} \leq \alpha^2\right\},\]
where $f_{m+1} = f$. Then $\hat{t} \geq \tilde{t}$ holds deterministically, and thus we have
\begin{align*}
    &\EEst{\ta_\D(f)^2}{\dt} \leq \EEst{\EEst{\frac{1}{\binom{r}{2}}\sum_{1 \leq j_1 < j_2 \leq r}\tilde{f}(X_{m+1,j_1})\tilde{f}(X_{m+1,j_2}) \cdot Z_{m+1,j_1}^{\tilde{t}}Z_{m+1,j_2}^{\tilde{t}}}{\D}}{\dt}\\
    &= \EEst{\frac{1}{\binom{r}{2}}\sum_{1 \leq j_1 < j_2 \leq r}\tilde{f}(X_{m+1,j_1})\tilde{f}(X_{m+1,j_2}) \cdot Z_{m+1,j_1}^{\tilde{t}}Z_{m+1,j_2}^{\tilde{t}}}{\dt}\\
    &= \frac{1}{m+1} \sum_{k=1}^{m+1} \EEst{\frac{1}{\binom{r}{2}}\sum_{1 \leq j_1 < j_2 \leq r}\tilde{f}_k(X_{k,j_1})\tilde{f}_k(X_{k,j_2}) \cdot Z_{k,j_1}^{\tilde{t}}Z_{k,j_2}^{\tilde{t}}}{\dt}\\
    &= \EEst{\frac{1}{m+1} \sum_{k=1}^{m+1} \frac{1}{\binom{r}{2}}\sum_{1 \leq j_1 < j_2 \leq r}\tilde{f}_k(X_{k,j_1})\tilde{f}_k(X_{k,j_2}) \cdot Z_{k,j_1}^{\tilde{t}}Z_{k,j_2}^{\tilde{t}}}{\dt},
\end{align*}
where we apply the fact that $\tilde{t}$ is a symmetric function of $(X_{k,1},\cdots,X_{k,r})_{k \in [m+1]}$. Then, by the definition of $\tilde{t}$, it follows that
\begin{align*}
    \EEst{\ta_\D(f)^2}{\dt} &\leq \alpha^2 \cdot \EEst{\frac{1}{m+1} \sum_{k=1}^{m+1} \frac{1}{\binom{r}{2}}\sum_{1 \leq j_1 < j_2 \leq r}\tilde{f}_k(X_{k,j_1})\tilde{f}_k(X_{k,j_2})}{\dt}\\
    &= \alpha^2 \cdot \EEst{\tilde{f}(X_{m+1,1})\tilde{f}(X_{m+1,2})}{\dt}.
\end{align*}
Applying the steps in~\eqref{eqn:pf_1}, we have $\EE{\ta_\D(f)^2} \leq \alpha^2$.

\end{document}